\documentclass[12pt]{article}

\usepackage{amsmath}
\usepackage{amssymb}
\usepackage{amsthm}
\usepackage{bm}
\usepackage{color}
\usepackage{fullpage}
\usepackage{graphicx}
\usepackage{natbib}
\usepackage{stmaryrd} 
\usepackage{subcaption}
\usepackage{thmtools}
\usepackage{tikz-cd} 
\usepackage[colorlinks,citecolor=blue,urlcolor=blue,filecolor=blue,backref=page]{hyperref}
\usepackage{bbold}
\usepackage{graphbox}

\newtheorem{definition}{Definition}
\newtheorem{principle}{Principle}
\newtheorem{theorem}{Theorem}
\newtheorem{corollary}{Corollary}

\newtheorem{example}{Example}

\newcommand{\R}{\mathbb{R}}

\usetikzlibrary{arrows}
\usetikzlibrary{calc}
\usetikzlibrary{shapes}
\usetikzlibrary{decorations.pathreplacing}

\begin{document}

\title{ A Role for Symmetry in the Bayesian Solution of Differential Equations}
\author{Junyang Wang$^1$, Jon Cockayne$^2$, Chris. J. Oates$^{1,2}$ \\
$^1$Newcastle University, UK \\
$^2$Alan Turing Institute, UK 
}
\maketitle

\begin{abstract}
The interpretation of numerical methods, such as finite difference methods for differential equations, as point estimators suggests that formal uncertainty quantification can also be performed in this context. 
Competing statistical paradigms can be considered and Bayesian probabilistic numerical methods (PNMs) are obtained when Bayesian statistical principles are deployed.
Bayesian PNM have the appealing property of being closed under composition, such that uncertainty due to different sources of discretisation in a numerical method can be jointly modelled and rigorously propagated.
Despite recent attention, no exact Bayesian PNM for the numerical solution of ordinary differential equations (ODEs) has been proposed.
This raises the fundamental question of whether exact Bayesian methods for (in general nonlinear) ODEs even exist.
The purpose of this paper is to provide a positive answer for a limited class of ODE.
To this end, we work at a foundational level, where a novel Bayesian PNM is proposed as a proof-of-concept. 
Our proposal is a synthesis of classical Lie group methods, to exploit underlying symmetries in the gradient field, and non-parametric regression in a transformed solution space for the ODE.
The procedure is presented in detail for first and second order ODEs and relies on a certain strong technical condition -- existence of a solvable Lie algebra -- being satisfied.
Numerical illustrations are provided.
\end{abstract}

\section{Introduction}

Numerical methods underpin almost all of scientific, engineering and industrial output.
In the abstract, a numerical task can be formulated as the approximation of a quantity of interest 
\begin{eqnarray*}
Q : \mathcal{Y} & \rightarrow & \mathcal{Q}, 
\end{eqnarray*}
subject to a finite computational budget. 
The true underlying state $\mathrm{y}^\dagger \in \mathcal{Y}$ is typically high- or infinite-dimensional, so that only limited information 
\begin{eqnarray}
A : \mathcal{Y} & \rightarrow & \mathcal{A} \label{eq: information operator}
\end{eqnarray}
is provided and exact computation of $Q(\mathrm{y}^\dagger)$ is prohibited.
For example, numerical integration aims to approximate an integral $Q(\mathrm{y}^\dagger) = \int \mathrm{y}^\dagger(t) \mathrm{d}t$ given the values $A(\mathrm{y}^\dagger) = \{(x_i,\mathrm{y}^\dagger(x_i))\}_{i=1}^n$ of the integrand $\mathrm{y}^\dagger$ on a finite number of abscissa $\{x_i\}_{i=1}^n$.
Similarly, a numerical approximation to the solution $Q(\mathrm{y}^\dagger) = \mathrm{y}^\dagger$ of a differential equation $\mathrm{d}\mathrm{y} / \mathrm{d}x = f(x,\mathrm{y}(x))$, $\mathrm{y}(x_0) = y_0$, must be based on at most a finite number of evaluations of $f$, the gradient field. 
In this viewpoint a numerical method corresponds to a map $b : \mathcal{A} \rightarrow \mathcal{Q}$, as depicted in Figure \ref{fig: diag1}, where $b(a)$ represents an approximation to the solution of the differential equation based on the information $a \in \mathcal{A}$.

The increasing ambition and complexity of contemporary applications is such that the computational budget can be \emph{extremely} small compared to the precision that is required at the level of the quantity of interest.
As such, in many important problems it is not possible to reduce the numerical error to a negligible level.
Fields acutely associated with this challenge include climate forecasting \citep{Wedi2014}, computational cardiology \citep{Chabiniok2016} and molecular dynamics \citep{Perilla2015}.
In the presence of non-negligible numerical error, it is unclear how scientific interpretation of the output of computation can proceed. 
\textcolor{black}{
For example, {\it a posteriori} analysis of traditional numerical methods can be used to establish hard upper bounds on the numerical error, but these bounds typically depend on an unknown global constant. 
In the case of ODEs, this may be the maximum value of a norm $\|f\|$ of the gradient field \citep[see e.g.][]{Estep1995}.
If $\|f\|$ were known, it would be possible to provide a hard bound on numerical error.
However, in the typical numerical context all that is known is that $\|f\|$ is finite. One could attempt to approximate $\|f\|$ with cubature, but that itself requires a numerical cubature method whose error is required to obey a known bound. 
In general, therefore, there are no hard error bounds without global information being {\it a priori} provided \citep{Larkin1974}.
Our aim in this paper is to consider, as an alternative to traditional numerical analysis, an exact Bayesian framework for solution uncertainty quantification in the ordinary differential equation context.
}

\subsection{Probabilistic Numerical Methods} \label{subsec: PNM}

The field of \emph{probabilistic numerics} dates back to \cite{Larkin1972} and a modern perspective is provided in \cite{Hennig2015,Oates2019}.
Under the abstract framework just described, numerical methods can be interpreted as point estimators in a statistical context, where the state $\mathrm{y}^\dagger$ can be thought of as a latent variable in a statistical model, and the `data' consist of information $A(\mathrm{y}^\dagger)$ that does not fully determine the quantity of interest $Q(\mathrm{y}^\dagger)$ but is indirectly related to it. 
\cite{Hennig2015} provide an accessible introduction and survey of the field. 
In particular, they illustrated how PNM can be used to quantify uncertainty due to discretisation in important scientific problems, such as astronomical imaging. 

Let the notation $\Sigma_{\mathcal{Y}}$ denote a $\sigma$-algebra on the space $\mathcal{Y}$ and let $\mathcal{P}_{\mathcal{Y}}$ denote the set of probability measures on $(\mathcal{Y},\Sigma_{\mathcal{Y}})$.
A probabilistic numerical method (PNM) is a procedure which takes as input a `belief' distribution $\mu \in \mathcal{P}_{\mathcal{Y}}$, representing epistemic uncertainty with respect to the true (but unknown) value $\mathrm{y}^\dagger$, along with a finite amount of information, $A(\mathrm{y}^\dagger) \in \mathcal{A}$.
The output is a distribution $B(\mu,A(\mathrm{y}^\dagger)) \in \mathcal{P}_{\mathcal{Q}}$ on $(\mathcal{Q},\Sigma_{\mathcal{Q}})$, representing epistemic uncertainty with respect to the quantity of interest $Q(\mathrm{y}^\dagger)$ after the information $A(\mathrm{y}^\dagger)$ have been processed. 
For example, a PNM for an ordinary differential equation (ODE) takes an initial belief distribution defined on the solution space of the differential equation, together with information arising from a finite number of evaluations of the gradient field, plus the initial condition of the ODE, to produce a distribution over either the solution space of the ODE, or perhaps some derived quantity of interest. 
In this paper, the measurability of $A$ and $Q$ will be assumed.

Despite computational advances in this emergent field, until recently there had not been an attempt to establish rigorous statistical foundations for PNM. 
In \cite{Cockayne2017} the authors argued that Bayesian principles can be adopted.
In brief, this framework requires that the \textcolor{black}{input belief distribution $\mu$ carries the semantics of a Bayesian agent's prior belief and the output of a PNM agrees with the inferences drawn when the agent is rational.}
\textcolor{black}{To be more precise recall that, in this paper, information is provided in a deterministic\footnote{\textcolor{black}{It is of course possible to perform Bayesian inference in the noisy-data context, but for the ODEs considered in this paper we assume that one can obtain noiseless evaluations of the gradient field.}} manner through \eqref{eq: information operator} and thus Bayesian inference corresponds to conditioning of $\mu$ on the level sets of $A$.} 
Let $Q_{\#} : \mathcal{P}_{\mathcal{Y}} \rightarrow \mathcal{P}_{\mathcal{Q}}$ denote the push-forward map associated to $Q$.
i.e. $Q_{\#}(\mu)(S) = \mu(Q^{-1}(S))$ for all $S \in \Sigma_{\mathcal{Q}}$.
Let $\{\mu^a\}_{a \in \mathcal{A}} \subset \mathcal{P}_{\mathcal{Y}}$ denote the disintegration, assumed to exist\footnote{The reader unfamiliar with the concept of a disintegration can treat $\mu^a$ as a technical notion of the `conditional distribution of $\mathrm{y}$ given $A(\mathrm{y}) = a$' when reading this work. The \emph{disintegration theorem}, Thm. 1 of \cite{Chang1997}, guarantees existence and uniqueness up to a $A_{\#}\mu$-null set under the weak requirement that $\mathcal{Y}$ is a metric space, $\Sigma_{\mathcal{Y}}$ is the Borel $\sigma$-algebra, $\mu$ is Radon, $\Sigma_{\mathcal{A}}$ is countable generated and $\Sigma_{\mathcal{A}}$ contains all singletons $\{a\}$ for $a \in \mathcal{A}$. }, of $\mu \in \mathcal{P}_{\mathcal{Y}}$ along the map $A$.

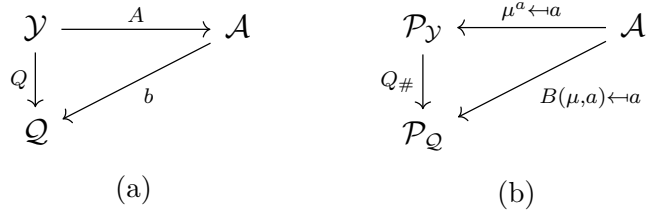
\begin{figure}[t!]
\centering
\begin{subfigure}{0.3\textwidth}
\centering
\begin{tikzcd}
\mathcal{Y} \arrow[d,"Q",swap] \arrow[rr,"A"] & & \mathcal{A} \arrow[lld,"b"] \\
\mathcal{Q}
\end{tikzcd}
\caption{}
\label{fig: diag1}
\end{subfigure}
\begin{subfigure}{0.3\textwidth}
\centering
\begin{tikzcd}
\mathcal{P}_{\mathcal{Y}} \arrow[d,"Q_{\#}",swap] & & \mathcal{A} \arrow[ll,"\mu^a \mapsfrom a",swap] \arrow[lld,"{B(\mu,a) \mapsfrom a}"] \\
\mathcal{P}_{\mathcal{Q}}
\end{tikzcd}
\caption{}
\label{fig: diag2}
\end{subfigure}
\caption{Diagrams for a numerical method.
(a) The traditional viewpoint of a numerical method is equivalent to a map $b$ from a finite-dimensional information space $\mathcal{A}$ to the space of the quantity of interest $\mathcal{Q}$.
(b) The probabilistic viewpoint treats approximation of $Q(\mathrm{y}^\dagger)$ in a statistical context, described by a map $B(\mu,\cdot)$ from $\mathcal{A}$ to the space of probability distributions on $\mathcal{Q}$.
The probabilistic numerical method $(A,B)$ is Bayesian if and only if (b) is a commutative diagram.}
\end{figure}

\begin{definition} \label{def: Bayesian}
A probabilistic numerical method $(A,B)$ with $A : \mathcal{Y} \rightarrow \mathcal{A}$ and $B : \mathcal{P}_{\mathcal{Y}} \times \mathcal{A} \rightarrow \mathcal{P}_{\mathcal{Q}}$ for a quantity of interest $Q : \mathcal{Y} \rightarrow \mathcal{Q}$ is \emph{Bayesian} if and only if $B(\mu,a) = Q_{\#}(\mu^a)$ for all $\mu \in \mathcal{P}_{\mathcal{Y}}$ and all $a \in \mathcal{A}$.
\end{definition}

This definition is intuitive; the output of the PNM should coincide with the marginal distribution for $Q(\mathrm{y}^\dagger)$ according to the disintegration element $\mu^a \in \mathcal{P}_{\mathcal{Y}}$, based on the information $a \in \mathcal{A}$ that was provided. 
The definition is equivalent to the statement that Figure \ref{fig: diag2} is a commutative diagram.
In \cite{Cockayne2017} the map $A$ was termed an \emph{information operator} and the map $B$ was termed a \emph{belief update operator}; we adhere to these definitions in our work.
The Bayesian approach to PNM confers several important benefits:
\begin{itemize}
\item The input $\mu$ and output $B(\mu,a)$ belief distributions can be interpreted, respectively, as a \emph{prior} and (marginal) \emph{posterior}.\footnote{Indeed, if the set $\mathcal{Y}^a=\{ \mathrm{y} \in \mathcal{Y}: A(\mathrm{y})=a \}$ is not measure zero under $\mu$, then $\mu^a$ is the conditional distribution defined by restricting $\mu$ to the subset $\mathcal{Y}^a$; $\mu^a(\mathrm{y}) = 1[\mathrm{y} \in \mathcal{Y}^a] \mu(\mathrm{y}) / \mu(\mathcal{Y}^a)$. 
The theory of disintegrations generalises the conditional distribution $\mu^a$  to cases where $\mathcal{Y}^a$ is a null set. 
}
As such, they automatically inherit the stronger formal semantics and philosophical foundations that underpin the Bayesian framework and, in this sense, are well-understood \citep[see e.g.][]{Gelman2013}.
\item The definition of Bayesian PNM is operational.
Thus, if we are presented with a prior $\mu$ and information $a$ then there is a unique Bayesian PNM and it is constructively defined.
\item The class of Bayesian PNM is closed under composition, such that uncertainty due to different sources of discretisation can be jointly modelled and rigorously propagated.
This point will not be discussed further in this work, but we refer the interested reader to Section 5 of \cite{Cockayne2017}.
\end{itemize}
Nevertheless, the strict definition of Bayesian PNM limits scope to design convenient computational algorithms and indeed several proposed PNM are not Bayesian \citep[see Table 1 in][]{Cockayne2017}.
The challenge is two-fold; for a Bayesian PNM, the elicitation of an appropriate prior distribution $\mu$ and the exact computation of its disintegration $\{\mu^a\}_{a \in \mathcal{A}}$ must both be addressed. 
In the next section we argue that -- perhaps as a consequence of these constraints -- a strictly Bayesian PNM for the numerical solution of an ODE does not yet exist.

\subsection{Existing Work for ODEs} \label{subsec: existing work}

The first PNM (of any flavour) for the numerical solution of ODEs, or which we are aware, was due to \cite{Skilling1992}.
Two decades later, this problem is receiving renewed critical attention as part of the active development of PNM.
The aim of this section is to provide a high-level overview of existing work and to argue that existing methods do not adhere to the definition of Bayesian PNM.

\paragraph{Notation:}
The notational convention used in this paper is that the non-italicised $\mathrm{y}$ denotes a generic function, whereas the italicised $y$ denotes a scalar value taken by the function $\mathrm{y}$.
The notation $\mathrm{y}^\dagger$ is reserved for the true solution to an ODE.
Throughout, the underlying state space $\mathcal{Y}$ is taken to be a space occupied by the true solution of the ODE, i.e. $\mathrm{y}^\dagger \in \mathcal{Y}$.

\subsubsection{\cite{Skilling1992}}
\label{subsubsec: Skilling}

The first paper on this topic, of which we are aware, was \cite{Skilling1992}.
This will serve as a prototypical PNM for the numerical solution of an ODE.
Originally described as `Bayesian' by the author, we will argue that, at least in the strict sense of Definition \ref{def: Bayesian}, it is not a Bayesian PNM.
Consider a generic univariate first-order initial value problem
\begin{eqnarray}
\frac{\mathrm{d}\mathrm{y}}{\mathrm{d}x} & = & f(x,\mathrm{y}(x)), \hspace{20pt} x \in [x_0,x_T], \hspace{20pt} \mathrm{y}(x_0) = y_0 . \label{eq: ODE1} 
\end{eqnarray}
Throughout this paper all ODEs that we consider will be assumed to be well-defined and admit a unique solution $\mathrm{y}^\dagger \in \mathcal{Y}$ where $\mathcal{Y}$ is some pre-specified set.
In this paper the quantity of interest $Q(\mathrm{y}^\dagger)$ will either be the solution curve $\mathrm{y}^\dagger$ itself or the value $\mathrm{y}^\dagger(x_T)$ of the solution at a specific input (in this section it will be the former).
The approach outlined in \cite{Skilling1992} allows for a general prior $\mu \in \mathcal{P}_{\mathcal{Y}}$.
The gradient field $f$ is treated as a `black box' oracle that can be queried at a fixed computational cost.
Thus we are provided with evaluations of the gradient field $\left[ f(x_0,y_0) , \dots , f(x_n,y_n) \right]^\top \in \mathbb{R}^{n+1}$ for certain input pairs $\{(x_i,y_i)\}_{i=0}^n$. 
 
This approach of treating evaluations of the gradient field as `data' will be seen to be a common theme in existing PNM for ODEs and theoretical support for this framework is rooted in the field of information-based complexity \citep{Traub1992}.
\color{black}
Let $a_i = f(x_i,y_i)$ and $a^i = [a_0,\dots,a_i]$.
The selection of the input pairs $(x_i,y_i)$ on which $f$ is evaluated is not constrained and several possibilities, of increasing complexity, were discussed in \cite{Skilling1992}.
To fix ideas, the simplest such approach is to proceed iteratively as follows:
\begin{enumerate}
\item[(0.1)] The first pair $(x_0,y_0)$ is fully determined by the initial condition of the ODE.
\item[(0.2)] The oracle then provides one piece of information, $a_0 = f(x_0,y_0)$.
\item[(0.3)] The prior $\mu$ is updated according to $a_0$, leading to a belief distribution $\mu_0$ which is just the disintegration element $\mu^{a^0}$.
\item[(1)] A discrete time step $x_1 = x_0 + h$, where $h = \frac{x_T - x_0}{n} > 0$, is performed and a particular point estimate $y_1 = \int \mathrm{y}(x_1) \mathrm{d} \mu_0(\mathrm{y})$ for the unknown true value $\mathrm{y}^\dagger(x_1)$ is obtained.
This specifies the second pair $(x_1,y_1)$.
\end{enumerate}
The process continues similarly, such that at time step $i-1$ we have a belief distribution $\mu_{i-1} = B(\mu,a^{i-1}) \in \mathcal{P}_{\mathcal{Y}}$, where the general belief update operator $B$ is yet to be defined, and the following step is performed:
\begin{enumerate}
\item[($i$)] Let $x_i = x_{i-1} + h$ and set $y_i = \int \mathrm{y}(x_i) \mathrm{d}\mu_{i-1}(\mathrm{y})$ .
\end{enumerate}
The final output is a probability distribution $\mu_n = B(\mu,a^n) \in \mathcal{P}_{\mathcal{Y}}$.
Now, strictly speaking, the method just described is \emph{not} a PNM in the concrete sense that we have defined.
Indeed, the final output $\mu_n$ is a deterministic function of the values $a^n$ of the gradient field that were obtained.
However, in the absence of additional assumptions on the global smoothness of the gradient field, the values of $f(x,y)$ outside any open neighbourhood of the true solution curve $\mathcal{C} = \{(x,y) : y = \mathrm{y}^\dagger(x), \; x \in [x_0,x_T]\}$ do not determine the solution of the ODE and, conversely, the solution of the ODE provides no information about the values of the gradient field outside any open neighbourhood of the true solution curve $\mathcal{C}$.
Thus it is not possible, in general, to write down an information operator $A : \mathcal{Y} \rightarrow \mathcal{A}$ that reproduces the information $a^n$ when applied to the solution curve $\mathrm{y}^\dagger(\cdot)$ of the ODE.

The approach taken in \cite{Skilling1992} was therefore to posit an \emph{approximate} information operator $\hat{A}$ and a particular belief update operator $B$, which are now described. 
The approximate information operator is motivated by the intuition that if $\mathrm{y}^\dagger(x_i)$ is well-approximated by $y_i$ at the abscissa $x_i$ then $\frac{\mathrm{d}\mathrm{y}^\dagger}{\mathrm{d}x}(x_i)$ should be well-approximated by $f(x_i,y_i)$.
That is, the following approximate information operator $\hat{A}$ was constructed:
\begin{eqnarray}
\hat{A}(\mathrm{y}) & = & \left[ \frac{\mathrm{d}\mathrm{y}}{\mathrm{d}x}(x_0) , \dots , \frac{\mathrm{d}\mathrm{y}}{\mathrm{d}x}(x_n) \right]^\top . \label{eq: Skilling B}
\end{eqnarray}
Of course, $\hat{A}(\mathrm{y}^\dagger) \neq a^n$ in general.
To acknowledge the approximation error, \cite{Skilling1992} proposed to model the information with an \textcolor{black}{approximate likelihood}:
\begin{eqnarray}
\frac{\mathrm{d}\mu_n}{\mathrm{d}\mu_0}(\mathrm{y}) & = & \prod_{i=1}^n\frac{\mathrm{d}\mu_i}{\mathrm{d}\mu_{i-1}}(\mathrm{y}) \\
\frac{\mathrm{d}\mu_i}{\mathrm{d}\mu_{i-1}}(\mathrm{y}) & \propto & \exp\left( -\frac{1}{2 \sigma^2} \left( \frac{\mathrm{d}\mathrm{y}}{\mathrm{d}x}(x_i) - f(x_i,y_i) \right)^2 \right)  \label{eq: Skilling Lhood}
\end{eqnarray}
This was referred to in \cite{Skilling1992} as simply a ``likelihood'' and, together with $\mu_0 = \mu^{a^0}$, the output $\mu_n$ is completely specified.
Here $\sigma$ is a fixed positive constant, however in principle a non-diagonal covariance matrix can also be considered.
\textcolor{black}{The negative consequences of basing inferences on an approximate information operator $\hat{A}$ are potentially twofold.
First, recall that values of the gradient field that are not contained on the true solution curve of the ODE do not, in principle, determine the true solution curve $\mathrm{y}^\dagger$.
It is therefore unclear if these values should be taken into account at all.
Second, in the special case where the gradient field $f$ does not depend the second argument then the quantities $\frac{\mathrm{d}\mathrm{y}}{\mathrm{d}x}(x_i)$ and $f(x_i,y_i)$ are identical.
From this perspective, $\mu_n$ represents inference under a mis-specified likelihood, since information is treated as erroneous when it is in fact exact.
The use of a mis-specified likelihood violates the \emph{likelihood principle} and implies violation of the Bayesian framework.
This confirms, through a different argument, that the approach of \cite{Skilling1992} cannot be Bayesian in the strict sense of Definition \ref{def: Bayesian}. }

\subsubsection{\cite{Schober2014,Schober2016,Teymur2016,Teymur2018}}

After \cite{Skilling1992}, several authors have proposed improvements to the above method.
The approach of \cite{Schober2014} considered Eq.~\eqref{eq: Skilling Lhood} in the $\sigma \downarrow 0$ limit.
In order that exact conditioning can be performed in this limit, the input belief distribution $\mu$ was restricted to be a $k$-times integrated Wiener measure on the solution space of the ODE.
The tractability of the integrated Weiner measure leads to a closed-form characterisation of the posterior and enables computation to be cast as a Kalman filter.

This direction of research can be motivated by the following fact:
For $k \in \{1,2\}$ the authors prove that if the input pair $(x_1,y_1)$ is  taken as $y_1 = \int \mathrm{y}(x_1) \mathrm{d}\mu_0(\mathrm{y})$, as indicated in Section \ref{subsubsec: Skilling}, then the smoothing estimate $\hat{y}_1 = \int \mathrm{y}(x_1) \mathrm{d} \mu_1(\mathrm{y})$, i.e. the posterior mean for $\mathrm{y}(x_1)$ based on information $a^1$, coincides with the deterministic approximation to $\mathrm{y}^\dagger(x_1)$ that would be provided by a $k$-th order Runge-Kutta method.
As such, theoretical guarantees such as local convergence order are inherited.
For $k = 3$ it was shown that the same conclusion can be made to hold, \emph{provided} that the input pair $(x_1,y_1)$ is selected in a manner that is no longer obviously related to $\mu_0$.
In all cases the identification with a classical Runge-Kutta method does not extend beyond iteration $n=1$.
Similar connections to multistep methods of Nordsieck and Adams form were identified, respectively, in \cite{Schober2016} and \cite{Teymur2016,Teymur2018}.
The approach of \cite{Schober2014} is not Bayesian in the sense of Definition \ref{def: Bayesian}, which can again be deduced from dependence on values of the gradient field away from the true solution curve, so that the likelihood principle is violated.

\subsubsection{\cite{Kersting2016}}

The subsequent work of \cite{Kersting2016} attempted to elicit an appropriate non-zero covariance matrix for use in Eq.~\eqref{eq: Skilling Lhood}, in order to encourage uncertainty estimates to be better calibrated.
Their proposal consisted of the \textcolor{black}{approximate likelihood}
\begin{eqnarray}
\frac{\mathrm{d}\mu_i}{\mathrm{d}\mu_{i-1}}(\mathrm{y}) & \propto & \exp\left( -\frac{1}{2} \left( \frac{ \frac{\mathrm{d}\mathrm{y}}{\mathrm{d}x}(x_i) - m_i }{\sigma_i} \right)^2 \right) \label{eq: Kersting Lhood} \\
m_i & = & \int f(x_i,\mathrm{y}(x_i)) \mathrm{d}\mu_{i-1}(\mathrm{y}) \label{eq: Kersting mean} \\
\sigma_i^2 & = & \int \left( f(x_i,\mathrm{y}(x_i)) - m_i \right)^2 \mathrm{d}\mu_{i-1}(\mathrm{y}) . \label{eq: Kersting var}
\end{eqnarray}
This can be viewed as the predictive marginal likelihood for the value $f(x_i,\mathrm{y}(x_i))$ based on $\mu_{i-1}$.
From a practical perspective, the approach is somewhat circular as the integrals in Eq.~\eqref{eq: Kersting mean} and \eqref{eq: Kersting var} involve the black-box gradient field $f$ and are therefore cannot be computed.
The authors suggested a number of ways that these quantities could be numerically approximated\footnote{One such method is \emph{Bayesian quadrature}, another PNM wherein the integrand $f$ is modelled as uncertain until it is evaluated. This raises separate philosophical challenges, as one must then ensure that the statistical models used for $\mathrm{y}(\cdot)$ and $f(x_i,\cdot)$ are logically consistent. In \cite{Kersting2016} these functions were simply modelled as independent.}, which involve evaluating $f(x_i,y_i)$ at one or more values $y_i$ that must be specified.
The overall approach again violates the likelihood principle and is therefore not Bayesian in the sense of Definition \ref{def: Bayesian}.

\subsubsection{\cite{Chkrebtii2013}}

The original work of \cite{Chkrebtii2013} is somewhat related to \cite{Kersting2016}, however instead of using the mean of the current posterior as input to the gradient field, the input pair $(x_i,y_i)$ was selected by sampling $y_i$ from the marginal distribution for $\mathrm{y}(x_i)$ implied by $\mu_{i-1}$.
The \textcolor{black}{approximate likelihood} in this approach was taken as follows:
\begin{eqnarray*}
\frac{\mathrm{d}\mu_i}{\mathrm{d}\mu_{i-1}}(\mathrm{y}) & \propto & \exp\left( -\frac{1}{2} \left( \frac{ \frac{\mathrm{d}\mathrm{y}}{\mathrm{d}x}(x_i) - f(x_i,y_i) }{\sigma_i} \right)^2 \right) \\ 
m_i & = & \int \frac{\mathrm{d}\mathrm{y}}{\mathrm{d}x}(x_i) \mathrm{d}\mu_{i-1}(\mathrm{y})  \\
\sigma_i^2 & = & \int \left( \frac{\mathrm{d}\mathrm{y}}{\mathrm{d}x}(x_i) - m_i \right)^2 \mathrm{d}\mu_{i-1}(\mathrm{y}) .
\end{eqnarray*}
Compared to Eq.~\eqref{eq: Kersting Lhood}, \eqref{eq: Kersting mean} and \eqref{eq: Kersting var}, this approach does not rely on integrals over the unknown gradient field.
However, the approach also relies on the approximate information operator in Eq.~\eqref{eq: Skilling B} and is thus not Bayesian according to Definition \ref{def: Bayesian}.

\subsubsection{\cite{Conrad2015,Abdulle2018}}

The approaches proposed in \cite{Conrad2015,Abdulle2018} are not motivated in the Bayesian framework, but instead seek to introduce a stochastic perturbation into a classical numerical method.
Both methods focus on the quantity of interest $Q(\mathrm{y}^\dagger) = \mathrm{y}^\dagger(x_T)$.
In the simple context of Eq.~\eqref{eq: ODE1}, the method of \cite{Conrad2015} augments the explicit Euler method with a stochastic perturbation:
\begin{eqnarray*}
y_i \; = \; y_{i-1} + h f(x_{i-1},y_{i-1}) + h^2 \epsilon_i, \hspace{20pt} x_i \; = \; x_{i-1} + h,  \hspace{20pt} i = 1,\dots,n 
\end{eqnarray*}
The distribution of the sequence $(\epsilon_i)_{i=1}^n$ must be specified. 
In the simplest case where the $\epsilon_i$ are modelled as independent, say with $\epsilon_i \sim \rho$, the canonical flow map $\Phi_i : \mathbb{R} \rightarrow \mathbb{R}$ of the explicit Euler method, defined as $\Phi_i(z) = z + h f(x_i,z)$, is replaced by the probabilistic counterpart $\Psi_i : \mathcal{P}_{\mathbb{R}} \rightarrow \mathcal{P}_{\mathbb{R}}$ given by 
$$
\Psi_i(\nu)(\mathrm{d}z) = \int \rho \left( \frac{\mathrm{d}z - \Phi_i(\tilde{z})}{h^2} \right) \nu(\mathrm{d}\tilde{z})
$$
through which stochasticity can be propagated.
The output of the method of \cite{Conrad2015} is then $B = \Psi_n \circ \dots \circ \Psi_1 \delta(y_0)$, where $\delta(z)$ denotes an atomic distribution on $z \in \mathbb{R}$.
For the case where each $\rho_i$ has zero mean, it can be shown that the mean of $B$ equals $\Phi_n \circ \dots \circ \Phi_1 (y_0)$, which is exactly the deterministic approximation produced with the explicit Euler method.

This framework can be practically problematic, since $\epsilon_i$ is charged with modelling the extrapolation error and such errors are not easily modelled as independent random variables -- Section 2.8 of \cite{Higham2002} is devoted to this point.
Thus, if for example $f(x,y) = y$, the true linearisation error at step $i$ is $e^{x_i} - e^{x_{i-1}}$ so that the `true' sequence $(\epsilon_i)_{i=1}^n$ in this case is monotonic and exponentially unbounded.
The challenge of designing a stochastic model for the sequence $(\epsilon_i)_{i = 1}^n$ that reflects the highly structured nature of the error remains unresolved.
On the other hand, the mathematical properties of this method are now well-understood \citep{Lie2018,Lie2017}.
The proposal of \cite{Abdulle2018} was to instead consider randomisation of the inputs $\{x_i\}_{i=0}^T$ in the context of a classical numerical method, also outside of the Bayesian framework.

\subsubsection{\cite{Cockayne2017,Tronarp2018}}

The survey just presented begs the question of whether a Bayesian PNM for ODEs can exist at all.
A first step toward this goal was taken in \cite{Cockayne2017}, where it was argued that an information operator can be constructed if the vector field $f$ is brought to the left-hand-side in Eq.~\eqref{eq: ODE1}.
Specifically, they proposed the information operator
\begin{eqnarray*}
\tilde{A}(\mathrm{y}) & = & \left[ \frac{\mathrm{d}\mathrm{y}}{\mathrm{d}x}(x_0) - f(x_0,\mathrm{y}(x_0)) , \dots , \frac{\mathrm{d}\mathrm{y}}{\mathrm{d}x}(x_n) - f(x_n,\mathrm{y}(x_n)) \right]^\top
\end{eqnarray*}
for which the `data' are trivial; $\tilde{a}^n = 0$.
It was rigorously established that the \textcolor{black}{approximate likelihood}
\begin{eqnarray*}
\frac{\mathrm{d}\mu_{i,\sigma}}{\mathrm{d}\mu_{i-1,\sigma}}(\mathrm{y}) & = & \exp\left( -\frac{1}{2 \sigma^2} \left( \frac{\mathrm{d}\mathrm{y}}{\mathrm{d}x}(x_i) - f(x_i,\mathrm{y}(x_i)) \right)^2 \right) 
\end{eqnarray*}
leads to an exact Bayesian PNM in the limit: $\mu_{n,\sigma} \overset{\mathcal{F}}{\rightarrow} \mu^{\tilde{a}^n}$ as $\sigma \downarrow 0$ for $\tilde{A}_{\#} \mu$-almost all $\tilde{a}^n \in \mathbb{R}^{n+1}$.
Here $\overset{\mathcal{F}}{\rightarrow}$ denotes convergence in an integral probability metric defined by a suitable set $\mathcal{F}$ of test functions \citep[see Sec. 4 of][]{Cockayne2017}.
However, the dependence of the information operator $\tilde{A}$ on $f$ means that this cannot be used as the basis for a practical method.
Indeed, unless $f$ depends linearly on its second argument and conjugacy properties of the prior can be exploited, the posterior cannot easily be characterised.
Approximate techniques from nonlinear filtering were proposed to address this challenge in \cite{Tronarp2018}.

\subsection{Our Contribution}

The comprehensive literature review in the previous section reveals not only that that no Bayesian PNM has yet been proposed, but also that such an endeavour may be fundamentally difficult.
Indeed, a theme that has emerged with existing PNM for ODEs, which can be traced back to \cite{Skilling1992}, is the use of approximate and subjective forms for the likelihood.
The complex, implicit relationship between the latent ODE solution $\mathrm{y}^\dagger$ and the data $f(x_i,y_i)$ arising from the gradient field appears to preclude use of an exact likelihood.
Of course, violation of the likelihood principle is not traditionally a concern in the design of a numerical method, yet if the strictly richer output that comes with a Bayesian PNM is desired, then clearly adherence to the likelihood principle is important.
It is therefore natural to ask the question, ``under what conditions can exact Bayesian inference for ODEs be made?''.

This paper presents a proof-of-concept PNM for the numerical solution of a (limited) class of ODEs that is both (a) Bayesian in the sense of Definition \ref{def: Bayesian} and (b) can in principle be implemented.
The method being proposed is indicated in Figure \ref{fig: schematic} and its main properties are as follows:
\begin{itemize}
\item The classical theory of Lie groups is exploited, for the first time in the context of PNM, to understand when an ODE of the form in Eq.~\eqref{eq: ODE1} can be transformed into an ODE whose gradient field is a function of the independent state variable only, reducing the ODE to an integral.
\item For ODEs that admit a solvable Lie algebra, our proposal can be shown to simultaneously perform exact Bayesian inference on both the original and the Lie-transformed ODE. 
Crucially, as we explain later, to identify a Lie algebra only high-level {\it a priori} information about the ODE is required.
The case of first- and second-order ODEs is presented in detail, but the method itself is general.
\item \textcolor{black}{In general the specification of prior belief can be difficult. 
The prior distributions that we construct are guaranteed to respect aspects of the structure of the ODE. 
As such, our priors are, to some extent, automatically adapted to the ODE at hand as opposed to being arbitrarily posited.}
\item In addition to the benefits conferred in the Bayesian framework, detailed in Section \ref{subsec: PNM} and in \cite{Cockayne2017}, the method being proposed can be computationally realised.
On the other hand, there is a cost in terms of the run-time of the method that is substantially larger than existing, non-Bayesian approaches \textcolor{black}{(especially classical numerical methods)}.
As such, we consider this work to be a proof-of-concept rather than an applicable Bayesian PNM.
\end{itemize}

\begin{figure}[t!]
\centering
\resizebox{.9\textwidth}{!}{
\begin{tikzpicture}
\node[anchor=south west,inner sep=0] (image) at (0,0) {
\includegraphics[width = 0.6\textwidth,clip,trim = 4cm 2cm 4cm 2cm,page=1]{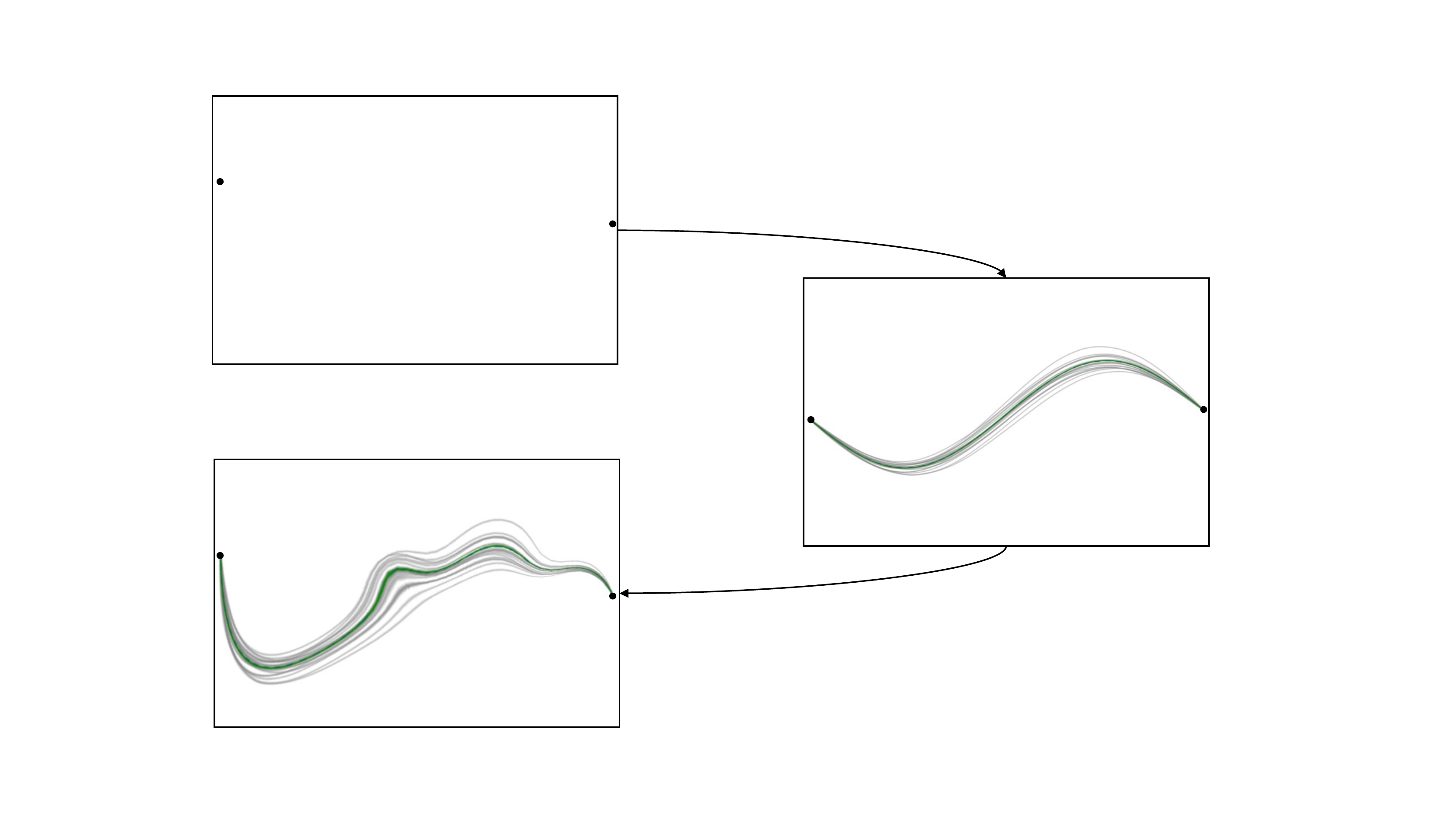}};
\node (A) at (-1.5,4.2) {$\cfrac{\mathrm{d}\mathrm{y}}{\mathrm{d}x} = f(x,\mathrm{y}(x))$};
\node (B) at (11,3) {$\cfrac{\mathrm{d}\mathrm{s}}{\mathrm{d}r} = g(r)$};
\node (C) at (8,5) {Lie transform; $(x,y) \mapsto (r,s)$};
\node (D) at (8.3,0.85) {(Lie transform)$^{-1}$; $(r,s) \mapsto (x,y)$};
\node (E) at (-1.3,1.7) {Exact};
\node (F) at (-1.3,1.3) {Bayesian};
\node (G) at (-1.3,0.9) {PNM};
\node (H) at (2,4.5) {?};
\end{tikzpicture} }
\caption{Schematic of our proposed approach.
An $n$th order ODE that admits a solvable Lie algebra can be transformed into $n$ integrals, to which exact Bayesian probabilistic numerical methods can be applied.
The posterior measure on the transformed space is then pushed back through the inverse transformation onto the original domain of interest.
}
\label{fig: schematic}
\end{figure}
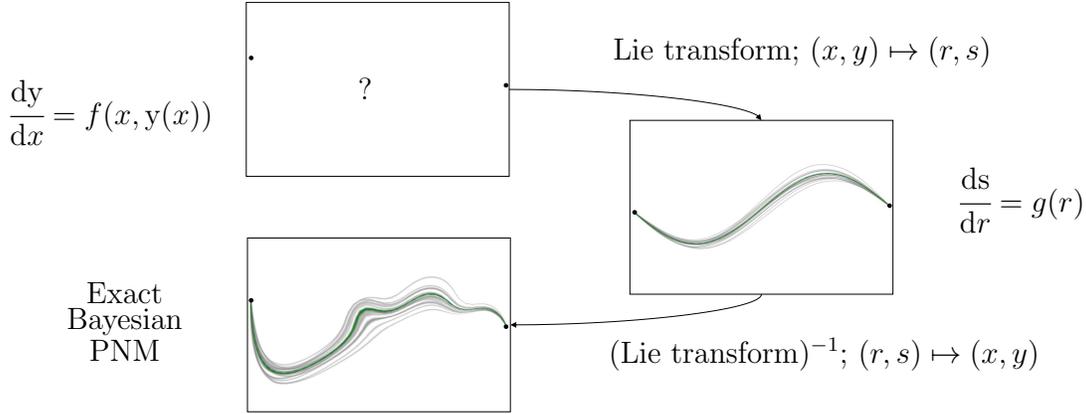

The remainder of the paper is structured as follows: Section \ref{sec: Lie} is dedicated to a succinct review of Lie group methods for ODEs. 
In Section \ref{sec: methods}, Lie group methods are exploited to construct priors over the solution space of the ODE whenever a solvable Lie algebra is admitted and exact Bayesian inference is performed on a transformed version of the original ODE which takes the form of an integral. 
Numerical experiments are reported in Section \ref{sec: experiment}.
Finally, some conclusions and recommendations for future research are drawn in Section \ref{sec: conclusion}.

\section{Background} \label{sec: Lie}

This section provides a succinct overview of classical Lie group methods, introduced in the 19th century by Sophus Lie in the differential equation context \citep{Hawkins2012}. 
Lie developed the fundamental notion of a \emph{Lie group of transformations}, which roughly correspond to maps that take one solution of the ODE to another. 
This provided a formal generalisation of certain algebraic techniques, such as dimensional analysis and transformations based on spatial symmetries, that can sometimes be employed to algebraically reduce the order of an ODE. 

This section proceeds as follows:
First, in Section \ref{subsec: lie transformations} we introduce a one-parameter Lie group of transformations and then, in Section \ref{subsec: invar under trans}, we explain what it means for a curve or a surface to be transformation-invariant.
In Secion \ref{subsec: symm method ODE} we recall consequences of Lie's theory in the ODE context.
Last, in Section \ref{subsec: multi-parameter} the generalisation to multi-parameter Lie groups is indicated.
Our development is heavily influenced by \cite{Bluman2002} and we refer the reader to their book when required.

\subsection{One-Parameter Lie Groups of Transformations}
\label{subsec: lie transformations}

The purpose of this section is to recall essential definitions, together with the \emph{first fundamental theorem of Lie}, which relates a Lie group of transformations to its infinitesimal generator.
In what follows we consider a fixed domain $D \subset \mathbb{R}^d$ and denote a generic state variable as $x = (x_1,\dots,x_d) \in D$.

\begin{definition}[One-Parameter Group of Transformations]
A \emph{one-parameter group of transformations} on $D$ is a map $X : D \times S \rightarrow D$, defined on $D \times S$ for some set $S \subset \mathbb{R}$, together with a bivariate map $\phi : S \times S \rightarrow S$, such that the following hold:
\begin{enumerate}
\item[(1)] For each $\epsilon \in S$, the transformation $X(\cdot , \epsilon)$ is a bijection on $D$.
\item[(2)] $(S , \phi)$ forms a group with law of composition $\phi$.
\item[(3)] If $\epsilon_0$ is the identity element in $(\textit{S}, \phi)$, then $X(\cdot,\epsilon_0)$ is the identity map on $D$.
\item[(4)] For all $x \in D$, $\epsilon,\delta \in S$, if $x^*=X(x,\epsilon)$, $x^{**}= X(x^*,\delta)$, then $x^{**}=X(x^*,\phi (\epsilon,\delta))$.
\end{enumerate}
\end{definition}

\noindent In what follows we continue to use the shorthand notation $x^* = X(x,\epsilon)$.
The notion of a \emph{Lie} group additionally includes smoothness assumptions on the maps that constitute a group of transformations.
Recall that a real-valued function is \emph{analytic} if it can be locally expressed as a convergent power series.

\begin{definition}[One-Parameter Lie Group of Transformations]
Let $X$, together with $\phi$, form a one-parameter group of transformations on $D$.
Then we say that $X$, together with $\phi$, form a \emph{one-parameter Lie group of transformations} on $D$ if, in addition, the following hold:
\begin{enumerate}
\item[(5)] $S$ is a (possibly unbounded) interval in $\mathbb{R}$. 
\item[(6)] For each $\epsilon \in S$, $X(\cdot,\epsilon)$ is infinitely differentiable in $\textit{D}$.
\item[(7)] For each $x \in D$, $X(x,\cdot)$ is an analytic function on $S$.
\item[(8)] $\phi$ is analytic in $S \times S$. 
\end{enumerate}
\end{definition}
\noindent Without the loss of generality it will be assumed, through re-parametrisation if required, that $S$ contains the origin and $\epsilon=0$ is the identity element in $(S, \phi)$.
The definition is illustrated through three examples:

\begin{example}[Translation in the x-Axis]\label{ex: translation}
The one-parameter transformation $x_1^* = x_1 + \epsilon$, $x_2^* = x_2$ for $\epsilon \in \mathbb{R}$ forms a Lie group of transformations on $D = \mathbb{R}^2$ with group composition law $\phi(\epsilon,\delta)=\epsilon+\delta$.
\end{example}

\begin{example}[Rotation Group] \label{ex: rotation}
The one-parameter transformation $x_1^* = x_1\cos(\epsilon)-x_2\sin(\epsilon)$, $x_2^* = x_1\sin(\epsilon)+x_2\cos(\epsilon)$ for $\epsilon \in \mathbb{R}$ again forms a Lie group of transformations on $D = \mathbb{R}^2$ with group composition law $\phi(\epsilon,\delta)=\epsilon+\delta$.
\end{example}

\begin{example}[Cyclic group $C_p$]
Let $D=\{1,2,3,\dots,p\}$. Let $S=\mathbb{Z}$. For $n\in D$ and $m\in S$, let $X(n,m)=n+m \; (\mathrm{mod}\, p)$.
Then $X$, together with $\phi(a,b)=a+b$, defines a one parameter group of transformations on $D$, but is not a Lie group of transformations since (5) is violated.
\end{example}

The first fundamental theorem of Lie establishes that a Lie group of transformations can be characterised by its infinitesimal generator, defined next:

\begin{definition}[Infinitesimal Transformation]
Let $X$ be a one-parameter Lie group of transformations. 
Then the transformation $x^* = x + \epsilon \xi(x)$,
\begin{eqnarray*}
\xi(x) & = & \left. \frac{\partial X(x,\epsilon)}{\partial \epsilon} \right|_{\epsilon = 0},
\end{eqnarray*}
is called the \emph{infinitesimal transformation} associated to $X$ and the map $\xi$ is called an \emph{infinitesimal}.
\end{definition}

\begin{definition}[Infinitesimal Generator]
The \emph{infinitesimal generator} of a one-parameter Lie group of transformations $X$ is defined to be the operator $\mathrm{X} = \xi \cdot \nabla$ where $\xi$ is the infinitesimal associated to $X$ and $\nabla=(\frac{\partial}{\partial x_1}, \frac{\partial}{\partial x_2}, \dots , \frac{\partial}{\partial x_n})$ is the gradient.
\end{definition}

\begin{example}[Ex. \ref{ex: translation}, continued]
For Ex. \ref{ex: translation}, we have 
\begin{eqnarray*}
\xi(x) & = & \left( \left. \frac{\mathrm{d}x_1^*}{\mathrm{d}\epsilon} \right|_{\epsilon = 0}, \left. \frac{ \mathrm{d} x_2^*}{\mathrm{d}\epsilon} \right|_{\epsilon = 0}\right) \; = \; (1,0)
\end{eqnarray*}
so the infinitesimal generator for translation in the x-axis is $\mathrm{X} = \frac{\partial}{\partial x_1}$.
\end{example}

\begin{example}[Ex. \ref{ex: rotation}, continued]
Similarly, the infinitesimal generator for the rotation group is $\mathrm{X} = -x_2\frac{\partial}{\partial x_1}+x_1\frac{\partial}{\partial x_2}$.
\end{example}

\noindent The first fundamental theorem of Lie provides a constructive route to obtain the infinitesimal generator from the transformation itself:

\begin{theorem}[First Fundamental Theorem of Lie; see pages 39-40 of \cite{Bluman2002}]
\label{thm: FTL 1}
A one parameter Lie group of transformations $X$ is characterised by the initial value problem:
\begin{eqnarray}
\frac{\mathrm{d}x^*}{\mathrm{d}\tau} \; = \; \xi(x^*), \hspace{20pt} x^* = x \text{ when } \tau = 0,  \label{eq: FTL 1} 
\end{eqnarray}
where $\tau(\epsilon)$ is a parametrisation of $\epsilon$ which satisfies $\tau(0) = 0$ and, for $\epsilon \neq 0$,
\begin{eqnarray}
\tau(\epsilon) & = & \int_0^{\epsilon} \cfrac{\partial \phi(a,b)}{\partial b} \Bigr|_{\substack{(a,b)=(\delta^{-1},\delta)}} \mathrm{d} \delta . \nonumber
\end{eqnarray}
Here $\delta^{-1}$ denotes the group inverse element for $\delta$.
\end{theorem}

\noindent Since Eq.~\eqref{eq: FTL 1} is translation-invariant in $\tau$, it follows that without loss of generality we can assume a parametrisation $\tau(\epsilon)$ such that the group action becomes $\phi(\tau_1,\tau_2) = \tau_1 + \tau_2$ and, in particular, $\tau^{-1} = - \tau$.
In the remainder of the paper, for convenience we assume that all Lie groups are parametrised such that the group action is $\phi(\epsilon_1,\epsilon_2)  = \epsilon_1 + \epsilon_2$.

The next result can be viewed as a converse to Theorem \ref{thm: FTL 1}, as it shows how to obtain the transformation from the infinitesimal generator.
All proofs are reserved for Supplemental Section \ref{proof section}.

\begin{theorem} \label{thm: exp generator}
A one parameter Lie group of transformations with infinitesimal generator $\mathrm{X}$ is equivalent to $x^* = e^{{\epsilon}\mathrm{X}} x$, where $e^{{\epsilon}\mathrm{X}}=\sum_{k=0}^{\infty} \frac{1}{k!}{\epsilon}^k\mathrm{X}^k x$.
\end{theorem}

\noindent The following is immediate from the proof of Theorem \ref{thm: exp generator}:

\begin{corollary} \label{cor: exp}
If $F$ is infinitely differentiable, then $F(x^*) = e^{\epsilon \mathrm{X}} F(x)$.
\end{corollary}

\subsection{Invariance Under Transformation}
\label{subsec: invar under trans}

In this section we explain what it means for a curve or a surface to be invariant under a Lie group of transformations and how this notion relates to the infinitesimal generator.

\begin{definition}[Invariant Function]
A function $F:D \rightarrow \mathbb{R}$ is said to be \emph{invariant} under a one parameter Lie group of transformations $x^* = X(x,\epsilon)$ if $F(x^*)=F(x)$ for all $x \in D$ and $\epsilon \in S$.
\end{definition}

\noindent Based on the results in Section \ref{subsec: lie transformations}, one might expect that invariance to a transformation can be expressed in terms of the infinitesimal generator of the transformation.
This is indeed the case:

\begin{theorem}\label{thm: invar 1}
A differentiable function $F:D \mapsto \mathbb{R}$ is invariant under a one parameter Lie group of transformations with infinitesimal generator $\mathrm{X}$ if and only if $\mathrm{X}F(x)=0$ for all $x \in D$.
\end{theorem}

\begin{theorem} \label{thm: invar 2}
For a function $F:D \mapsto \mathbb{R}$ and a one parameter Lie group of transformations $x^* = X(x,\epsilon)$, the relation $F(x^*) = F(x)+\epsilon$ holds for all $x \in D$ and $\epsilon \in S$ if and only if $\mathrm{X}F(x)=1$ for all $x \in D$.
\end{theorem}

The following definition is fundamental to the method proposed in Section \ref{sec: methods}:

\begin{definition}[Canonical Coordinates] \label{def: canonical}
Consider a coordinate system $r = (r_1(x),  \dots , r_n(x))$ on $D$.
Then any one parameter Lie group of transformations $x^* = X(x,\epsilon)$ induces a transformation of the coordinates $r_i^* = r_i(x^*)$.
The coordinate system $r$ is called \emph{canonical} for the transformation if $r^*_1 = r_1 , \dots , r^*_{n-1} = r_{n-1}$ and $r^*_n = r_n + \epsilon$.
\end{definition}

\begin{example}[Ex. \ref{ex: rotation}, continued] 
For the rotation group in Ex. \ref{ex: rotation}, we have canonical coordinates $r_1(x_1,x_2)=\sqrt{x_1^2+x_2^2}$ , $r_2(x_1,x_2)=\mathrm{arctan}(x_2/x_1)$.
\end{example}

\noindent In canonical coordinates, a one parameter Lie group of transformations can be viewed as a straight-forward translation in the $r_n$-axis.
The existence of canonical coordinates is established in Thm. 2.3.5-2 of \cite{Bluman2002}.
Note that Thms. \ref{thm: invar 1} and \ref{thm: invar 2} imply that $\mathrm{X}r^*_i=0$ for $i=1,2,...,n-1$, $\mathrm{X}r^*_n=1$.

\begin{definition}[Invariant Surface]
For a function $F:D \rightarrow \mathbb{R}$, a surface defined by $F(x)=0$ is said to be \emph{invariant} under a one parameter Lie group of transformation $x^* = X(x,\epsilon)$ if and only if $F(x^*)=0$ whenever $F(x)=0$ for all $x \in D$ and $\epsilon \in S$.
\end{definition}

\noindent The invariance of a surface, as for a function, can be cast in terms of an infinitesimal generator:

\begin{corollary} 
A surface $F(x)=0$ is invariant under a one parameter Lie group of transformations with infinitesimal generator $\mathrm{X}$ if and only if $\mathrm{X}F(x)=0$ whenever $F(x)=0$.
\end{corollary}

\subsection{Symmetry Methods for ODEs}
\label{subsec: symm method ODE}

The aim of this section is to relate Lie transformations to ODEs for which these transformations are admitted.
These techniques form the basis for our proposed method in Section \ref{sec: methods}.

For an ODE of the form in Eq.~\eqref{eq: ODE1}, one can consider the action of a transformation on the coordinates $(x,y)$; i.e. a special case of the above framework where the generic coordinates $x_1$ and $x_2$ are respectively the independent ($x$) and dependent ($y$) variables of the ODE.
It is clear that such a transformation also implies some kind of transformation of the derivatives $y_m := \frac{\mathrm{d}^m y}{\mathrm{d} x^m}$.
Indeed, consider a one-parameter Lie group of transformations $(x^*,y^*) = (X(x,y;\epsilon), Y(x,y;\epsilon))$.
Then we have from the chain rule that $y^*_m := \frac{\mathrm{d}^m y^*}{\mathrm{d}(x^*)^m}$ is a function of $x,y,y_1, \dots ,y_m$ and we denote $y_m^* = Y_m(x,y,y_1, \dots ,y_m;\epsilon)$.
As an explicit example:
\begin{eqnarray*}
y^*_1 & = & \frac{\mathrm{d}y^*}{\mathrm{d}x^*} \; = \; \frac{\frac{\partial{Y(x,y;\epsilon)}}{\partial x}+y_1\frac{\partial{Y(x,y;\epsilon)}}{\partial y}}{\frac{\partial{X(x,y;\epsilon)}}{\partial x}+y_1\frac{\partial{X(x,y;\epsilon)}}{\partial y}} \; =: \; Y_1(x,y,y_1;\epsilon)
\end{eqnarray*}
In general:
\begin{eqnarray*}
y^*_m & = & \frac{\frac{\partial{y^*_{m-1}}}{\partial x}+y_1\frac{\partial{y^*_{m-1}}}{\partial y}+y_2\cfrac{\partial{y^*_{m-1}}}{\partial y_1}+...+y_m\frac{\partial{y^*_{m-1}}}{\partial y_{m-1}}}{\frac{\partial{X(x,y;\epsilon)}}{\partial x}+y_1\frac{\partial{X(x,y;\epsilon)}}{\partial y}} \; =: \; Y_m(x,y,y_1, \dots ,y_m;\epsilon)
\end{eqnarray*}
In this sense a transformation defined on $(x,y)$ can be naturally extended to a transformation on $(x,y,y_1,y_2,\dots)$ as required.

\begin{definition}[Admitted Transformation]
An $m$th order ODE $F(x,y,y_1, \dots , y_m)=0$ is said to \emph{admit} a one parameter Lie group of transformations $(x^*,y^*)=(X(x,y;\epsilon), Y(x,y;\epsilon))$ if the surface $F$ defined by the ODE is invariant under the Lie group of transformations, i.e. if $F(x^*,y^*,y_1^*, \dots ,y_m^*)=0$ whenever $F(x,y,y_1, \dots , y_m) = 0$.
\end{definition}

\begin{example}
Clearly any ODE of the form $\frac{\mathrm{d}y}{\mathrm{d}x}=F(x)$ admits the transformation $(x^*,y^*)=(x, y+\epsilon)$.
\end{example}

Our next task is to understand how the infinitesimal generator of a transformation can be extended to act on derivatives $y_m$.

\begin{definition}[Extended Infinitesimal Transformation]\label{extendedeta}
The $m$th \emph{extended infinitesimals} of a one parameter Lie group of transformations $(x^*,y^*) = (X(x,y;\epsilon) , Y(x,y;\epsilon) )$ are defined as the functions $\xi, \eta, \eta^{(1)} , \dots , \eta^{(m)}$ for which the following equations hold:
\begin{align*}
x^* & =  X(x,y;\epsilon) & = & \; x+\epsilon\xi(x,y)+O(\epsilon^2) \\
y^* & = Y(x,y;\epsilon) & = & \; y+\epsilon\eta(x,y)+O(\epsilon^2) \\
y^*_1 & = Y_1(x,y,y_1;\epsilon) & = & \; y_1+\epsilon\eta^{(1)}(x,y,y_1)+O(\epsilon^2) \\
& \vdots & & \\
y^*_m & = Y_m(x,y,y_1, \dots ,y_m;\epsilon) & = & \; y_m+\epsilon\eta^{(m)}(x,y,y_1,y_2, \dots , y_m)+O(\epsilon^2)
\end{align*}
\end{definition}
\noindent It can be shown straightforwardly via induction that 
\begin{eqnarray}
\eta^{(m)}(x,y,y_1,y_2, \dots , y_m) & = & \frac{\mathrm{d}^m\eta}{\mathrm{d}x^m}-\sum_{k=0}^{m}\frac{m!}{(m-k)!k!}y_{m-k-1}\frac{\mathrm{d}^k\xi}{\mathrm{d}x^k} \label{eq: extended infinitesimals}
\end{eqnarray}
where $\frac{\mathrm{d}}{\mathrm{d}x}$ denotes the full derivative with respect to $x$, i.e. $\frac{\mathrm{d}}{\mathrm{d}x}=\frac{\partial}{\partial x}+y_1\frac{\partial}{\partial y}+\sum_{k=2}^{m+1}y_k\frac{\partial}{\partial y_{k-1}}$.
It follows that $\eta^{(m)}$ is a polynomial in $y_1,y_2,\dots,y_m$ with coefficients linear combinations of $\xi, \eta$ and their partial derivatives up to the $m$th order.

\begin{definition}[Extended Infinitesimal Generator]
The $m$th \emph{extended infinitesimal generator} is defined as
\begin{eqnarray*}
\mathrm{X}^{(m)} & = & \xi_m(x,y,y_1, \dots ,y_m)\cdot\nabla \\
& = & {\xi}(x,y)\frac{\partial}{\partial x}+{\eta}(x,y)\frac{\partial}{\partial y}+{\eta^{(1)}}(x,y)\frac{\partial}{\partial y_1} + \dots +{\eta^{(m)}}(x,y,y_1, \dots ,y_m)\frac{\partial}{\partial y_m} 
\end{eqnarray*}
where $\nabla=(\frac{\partial}{\partial x}, \frac{\partial}{\partial y}, \frac{\partial}{\partial y_1} , \dots , \frac{\partial}{\partial y_m})$ is the extended gradient.
\end{definition}

\noindent The following corollaries are central to the actual computation of the admitted Lie groups of an ODE.

\begin{corollary}\label{thm: invar 3}
A differentiable function $F:D_m \rightarrow \mathbb{R}$ where $D_m$ is the phase space containing elements of the form $(x,y,y_1,\dots,y_m)$, is invariant under a one parameter Lie group of transformations with an extended infinitesimal generator $\mathrm{X}^{(m)}$ if and only if $\mathrm{X}^{(m)}F(x,y,y_1,\dots,y_m)=0$ for all $(x,y,y_1,\dots,y_m) \in D_m$.
\end{corollary}

\begin{corollary}[Infinitesimal Criterion for Symmetries Admitted by an ODE] \label{thm: Invar criteria}
A one parameter Lie group of transformations is admitted by the $m$th order ODE $F(x,y,y_1,..,y_m)=0$ if and only if its extended infinitesimal generator $\mathrm{X}^{(m)}$ satisfies $\mathrm{X}^{(m)}F(x,y,y_1,\dots ,y_m)=0$ whenever $F(x,y,y_1,\dots,y_m)=0$.
\end{corollary}

\subsection{Multi-Parameter Lie Groups and Lie Algebras}
\label{subsec: multi-parameter}

To leverage the full power of Lie symmetry methods for ODEs of order $m \geq 2$, we need to consider multiple Lie symmetries which are collectively described by a \emph{Lie algebra}. 
Fortunately, the notion of a multi-parameter Lie group of transformations is a natural generalisation from the one parameter case.
Thus, this last section of background material concerns the generalisation of the definitions in Section \ref{subsec: lie transformations} to the case of a multi-parameter Lie group.
The associated Lie algebra will also be defined.

\begin{definition}[Multi-Parameter Lie Group of Transformations]
The set of transformations $x^* = X(x,\epsilon)$ where $x_i^* = X_i(x , \epsilon)$ and $\epsilon = (\epsilon_1,\epsilon_2, \dots ,\epsilon_r) \in S \subset {\R}^r$ is called a \emph{$r$-parameter Lie group of transformations} if it satisfies the same axioms as in the one parameter case, but with law of composition $\phi(\epsilon,\delta)=(\phi_1(\epsilon,\delta), \dots ,\phi_r(\epsilon,\delta))$, and (without loss of generality) $\epsilon=(0,0,...,0)$ as the group identity element.
\end{definition}

\begin{definition}[Infinitesimal Matrix]
The appropriate generalisation for the infinitesimal transformation is the infinitesimal matrix $\Xi = [\xi_{ij}]$, whose entries are defined as $\xi_{ij}(x) = \frac{\partial X_j(x,\epsilon)}{\partial \epsilon_i}\Bigr|_{\substack{\epsilon=0}}$.
\end{definition}

\begin{definition}[Infinitesimal Generator]
An $r$-parameter Lie group of transformations is associated with $r$ infinitesimal generators, $\mathrm{X_i}$, defined as $\mathrm{X}_i  = \mathrm{X}_i(x) = \sum_{j=1}^{d} \xi_{ij}(x)\frac{\partial}{\partial x_j}$.
\end{definition}

The first fundamental theorem of Lie can be generalised to the multi-parameter case. 
In particular, it can be shown that an $r$-parameter Lie group of transformations is characterised by the set of its $r$ infinitesimal generators.
The generalisation is straight-forward and so, for brevity, we refer the reader to pages 39-40 of \cite{Bluman2002}.

Next we explain how the collection of infinitesimal generators forms a Lie algebra.
This relies on the basic facts that the set $\mathcal{D}$ of differential operators on $D$ is a vector space over $\R$ (i.e. $\lambda \mathrm{X} + \mu \mathrm{Y} \in \mathcal{D}$ for all $\mathrm{X} , \mathrm{Y} \in \mathcal{D}$ and $\lambda , \mu \in \mathbb{R}$) and that differential operators can be composed (i.e. $\mathrm{XY} \in \mathcal{D}$ for all $\mathrm{X}, \mathrm{Y} \in \mathcal{D}$).

\begin{definition}[Commutator] \label{eq: commutator}
The \emph{commutator} of two infinitesimal generators $\mathrm{X}_i$ and $\mathrm{X}_j$ is defined as $[\mathrm{X}_i,\mathrm{X}_j] = \mathrm{X}_i\mathrm{X}_j-\mathrm{X}_j\mathrm{X}_i$. 
\end{definition}

\begin{theorem}[Second Fundamental Theorem of Lie; see page 78 of \cite{Bluman2002}] \label{thm: 2nd fund}
Consider an $r$-parameter Lie group of transformations and let $\mathcal{L}$ denote the linear span of the infinitesimal generators $\mathrm{X}_1,\dots,\mathrm{X}_r$ in $\mathcal{D}$.
Let $[ \cdot , \cdot] : \mathcal{L} \times \mathcal{L} \rightarrow \mathcal{D}$ denote the unique bilinear operator that agrees with Def.~\eqref{eq: commutator} on the set of infinitesimal generators.
i.e. 
\begin{eqnarray}
\left[ \sum_{i=1}^r \lambda_i \mathrm{X}_i , \sum_{j=1}^r \mu_j \mathrm{X}_j \right] & = & \sum_{i=1}^r \sum_{j=1}^r \lambda_i \mu_j (\mathrm{X}_i \mathrm{X}_j - \mathrm{X}_j \mathrm{X}_i) . \label{eq: linear extension}
\end{eqnarray}
Then $[ \cdot , \cdot ]$ maps into $\mathcal{L}$. i.e. the right hand side of Eq.~\eqref{eq: linear extension} belongs to $\mathcal{L}$ for all $\lambda , \mu \in \R^r$.
\end{theorem}

\begin{example} 
Consider the two parameter group of transformations on $D = \mathbb{R}^2$ given by $(x^*,y^*)=(x+x\epsilon+x^2\delta, y+y\epsilon+y^2\delta)$.
The infinitesimal generators corresponding to $\delta$ and $\epsilon$, respectively, are $\mathrm{X}_1 = x^2\frac{\partial}{\partial x}+y^2\frac{\partial}{\partial y}$, $\mathrm{X}_2 = x\frac{\partial}{\partial x}+y\frac{\partial}{\partial y}$.
It can be directly verified that $[\mathrm{X}_1,\mathrm{X}_2]=-\mathrm{X}_1$.
\end{example}

The space $\mathcal{L}$, defined in Thm. \ref{thm: 2nd fund}, satisfies the properties of an $r$-dimensional Lie algebra $\mathcal{L}$, defined next:

\begin{definition}[Lie Algebra]
An $r$-dimensional vector space $\mathcal{L}$ over $\R$ together with a bilinear operator $[ \cdot , \cdot ] : \mathcal{L} \times \mathcal{L} \rightarrow \mathcal{L}$ is called an $r$-dimensional \emph{Lie algebra} if the following hold:
\begin{enumerate}
\item[(1)] Alternativity: $[\mathrm{X},\mathrm{X}]=0$ for all $\mathrm{X} \in \mathcal{L}$
\item[(2)] Jacobi Identity: $[\mathrm{X},[\mathrm{Y},\mathrm{Z}]]+[\mathrm{Y},[\mathrm{Z},\mathrm{X}]]+[\mathrm{Z},[\mathrm{X},\mathrm{Y}]]=0$ for all $\mathrm{X},\mathrm{Y},\mathrm{Z} \in \mathcal{L}$
\end{enumerate}
\end{definition}

In general, for the methods presented in Section \ref{sec: methods} to be applied, existence of an $n$-parameter Lie group of transformations is not in itself sufficient; we require the existence of an $n$-dimensional \emph{solvable} Lie sub-algebra, defined next:

\begin{definition}[Normal Lie Sub-algebra]
Consider a Lie sub-algebra $\mathcal{J}$ of a Lie algebra $\mathcal{L}$ with bilinear operator $[\cdot , \cdot]$, i.e. a subset $\mathcal{J} \subset \mathcal{L}$ such that, when equipped with the restriction of $[\cdot , \cdot]$ to $\mathcal{J} \times \mathcal{J}$, is itself a Lie algebra and, in particular, $[\mathrm{X},\mathrm{Y}] \in \mathcal{J}$ for all $\mathrm{X} , \mathrm{Y} \in \mathcal{J}$.
Then $\mathcal{J}$ is said to be \emph{normal} if, in addition, $[\mathrm{X},\mathrm{Y}] \in \mathcal{J}$ for all $\mathrm{X} \in \mathcal{J}, \mathrm{Y} \in \mathcal{L}$.
\end{definition}

\begin{definition}[Solvable Lie Algebra]
An $r$-dimensional Lie algebra $\mathcal{L}$ is called \emph{solvable} if there exists a chain of sub-algebras $\mathcal{L}^1 \subset \mathcal{L}^2 \subset...\subset \mathcal{L}^{q-1} \subset \mathcal{L}^r =: \mathcal{L}$ such that $\mathcal{L}^{i-1}$ is a normal sub-algebra of $\mathcal{L}^{i}$ for $i=2,3,...,r$.
\end{definition}


For low-order ODEs, the existence requirement for an admitted Lie group of transformations is more restrictive than the requirement that the associated Lie algebra is solveable. 
Indeed, we have the following result:

\begin{theorem} \label{thm: 2dsolvable}
All two-dimensional Lie Algebras are solvable.
\end{theorem}

This completes our review of background material.
The exact Bayesian PNM developed in Section \ref{sec: methods} for an $n$th order ODE require the existence of an admitted $n$-parameter Lie group of transformations with a solvable Lie algebra.
In practice we therefore require some high-level information on the gradient field $f$, in order to establish which transformations of the ODE may be admitted.
In addition, the requirement of a solvable Lie algebra also limits the class of ODEs for which our exact Bayesian methods can be employed.
Nevertheless, this class of ODEs is sufficiently broad to have merited extensive theoretical research \citep{Bluman2002} and the development of software \citep{Baumann2013}.

\section{Methods} \label{sec: methods}

In this section our novel Bayesian PNM is presented.
The method relies on high-level information about the gradient field $f$ and, in Section \ref{subsec: ode to symmetries}, we discuss how such information can be exploited to identify any Lie transformations that are admitted by the ODE.
In the case of a first order ODE, any non-trivial  transformation is sufficient for our method and an explicit information operator is provided for this case, together with recommendations for prior construction, in Section \ref{subsec: prior 1st order}.
Together, the prior and the information operator uniquely determine a Bayesian PNM, as explained in Section \ref{subsec: PNM}.
In the general case of an $m$th order ODE, we require that an $m$-dimensional solvable Lie algebra is admitted by the ODE.
The special case $m = 2$ is treated in detail, with an explicit information operator and guidance for prior construction provided in Section \ref{subsec: prior 2nd order}.
In the Supplemental Section \ref{subsec: training set} the selection of input pairs $(x_i,y_i)$ to the gradient field is discussed.

\subsection{From an ODE to its Admitted Transformations}
\label{subsec: ode to symmetries}

For the methods proposed in this paper, transformations admitted by the ODE, together with their infinitesimal generators, must first be obtained.
The algorithm for obtaining infinitesimal generators follows as a consequence of Cor. \ref{thm: Invar criteria}. 
Indeed, suppose we have a $m$th order ODE of the form $y_m-f(x,y,y_1,\dots,y_{m-1})=0$.
Then, by Cor. \ref{thm: Invar criteria}, a transformation with infinitesimal generator $\mathrm{X}$ is admitted by the ODE if and only if:
\begin{equation}\label{Invar equation}
\mathrm{X}^{(m)}(y_m-f(x,y,y_1,\dots,y_{m-1}))=0 \quad \mathrm{whenever} \quad y_m=f(x,y,y_1,\dots,y_{m-1}) .
\end{equation}
In infinitesimal notation, Eq.~\eqref{Invar equation} is equivalent to 
\begin{equation}\label{Invar equation2}
\eta^{(m)}(x,y,y_1,\dots,y_{m-1},y_m)=\xi\frac{\partial f}{\partial x}+\eta\frac{\partial f}{\partial y}+\sum_{k=1}^{m-1}\eta^{(k)}\frac{\partial f}{\partial y_k} .
\end{equation}
The direct solution of Eq.~\eqref{Invar equation2} recovers any transformations that are admitted.

In the common scenario where $f(x,y,y_1,\dots,y_{m-1})$ is a polynomial in $y_1,y_2,\dots,y_{m-1}$, the algorithm just described, for identification of admitted transformations, can be fully automated \citep[c.f.][]{Baumann2013}.
Indeed, from Def. \ref{extendedeta} it follows that the extended infinitesimals $\eta^{(k)}$ for $k \in 1,2,3,\dots,m$ are polynomial in $y_1,y_2,\dots,y_{k}$. 
Thus, by substituting $y_m=f(x,y,y_1,\dots,y_{m-1})$, Eq.~\eqref{Invar equation} too must be a polynomial in $y_1,y_2,\dots,y_{m-1}$.
Moreover, the coefficients of this polynomial must vanish because \eqref{Invar equation} holds for arbitrary values of $x,y,y_1,\dots,y_{m-1}$. 
This argument, of setting coefficients to zero, leads to a system of linear partial differential equations (overdetermined when $m \geq 2$) for $\xi(x,y)$ and $\eta(x,y)$, which can be exactly solved to retrieve all the infinitesimal generators of the ODE.
The same strategy can often be applied beyond the polynomial case and explicit worked examples of this procedure are now provided:

\begin{example}[First Order ODE]
\label{1storderexample initial}
Consider the class of all first order ODEs of the form $\frac{\mathrm{d}\mathrm{y}}{\mathrm{d}x} = f(x,\mathrm{y}(x))$, $f(x,y) = F\left( \frac{y}{x} \right)$.
From Eq.~\eqref{eq: extended infinitesimals}, we have $\eta^{(1)} = \eta_x+(\eta_y-\xi_x)y_1-\xi_y(y_1)^2$ so Eq.~\eqref{Invar equation2} becomes $\eta_x+(\eta_y-\xi_x)f-\xi_y(f)^2 = \xi\frac{\partial f}{\partial x}+\eta\frac{\partial f}{\partial y}$ and thus $\eta_x+(\eta_y-\xi_x)F\left(\frac{y}{x}\right)-\xi_yF\left(\frac{y}{x}\right)^2 = -\xi F'\left(\frac{y}{x}\right)\frac{y}{x^2}+\eta F'\left(\frac{y}{x}\right)\frac{1}{x}$.
For this equation to hold for general $F$, the coefficients of $F$, $F^2$ and $F'$ must vanish: $\eta_x = 0$, $\eta_y-\xi_x = 0$, $\xi_y = 0$, $-\xi\frac{y}{x^2}+\eta\frac{1}{x} = 0$.
This is now a linear system of partial differential equations in $(\xi,\eta)$ which is easily solved; namely $\xi=x, \eta=y$.
The associated infinitesimal generator is $\mathrm{X}=x\frac{\partial}{\partial x}+y\frac{\partial}{\partial y}$. 
\end{example}

\begin{example}[Second Order ODE]\label{2ndorderexample0}
The infinitesimal generators for the second order ODE
\begin{equation} \label{eq:2ndorderexample}
(x-\mathrm{y}(x))\frac{\mathrm{d}^2\mathrm{y}}{\mathrm{d}x^2}+2\frac{\mathrm{d}\mathrm{y}}{\mathrm{d}x} \left( \frac{\mathrm{d}\mathrm{y}}{\mathrm{d}x}+1 \right) +\left(\frac{\mathrm{d}\mathrm{y}}{\mathrm{d}x} \right)^{3/2} = 0
\end{equation}
are derived in Supplementary Section \ref{proof section}.
\end{example}

\subsection{The Case of a First Order ODE}
\label{subsec: prior 1st order}

In this section we present our approach for a first order ODE.
This allows some of the more technical details associated to the general case to be omitted, due to the fact that any one-dimensional Lie algebra is trivial.
The main result that will allow us to construct an exact Bayesian PNM is as follows:

\begin{theorem}[Reduction of a First Order ODE to an Integral]\label{1stordertoquad}
If a first order ODE 
\begin{eqnarray}\label{1storder}
\frac{\mathrm{d}\mathrm{y}}{\mathrm{d}x} & = & f(x,\mathrm{y}(x))
\end{eqnarray}
admits a one parameter Lie group of transformations, then there exists coordinates $r(x,y)$, $s(x,y)$ such that 
\begin{eqnarray}
\frac{\mathrm{d}\mathrm{s}}{\mathrm{d}r} & = & G(r) \label{eq: G integral}
\end{eqnarray}
for some explicit function $G(r)$.
\end{theorem}

Note that the transformed ODE in Eq.~\eqref{eq: G integral} is nothing more than an integral, for which exact Bayesian PNM have already been developed \citep[e.g.][]{Briol2016,Karvonen2018}.
At a high level, as indicated in Fig. \ref{fig: schematic}, our proposed Bayesian PNM performs inference for the solution $\mathrm{s}(r)$ of Eq.~\eqref{eq: G integral} and then transforms the resultant posterior back into the original $(x,y)$-coordinate system.
Our PNM is therefore based on the information operator
\begin{eqnarray}
A(\mathrm{y}) & = & \left[ G(r_0) , \dots , G(r_n) \right]^\top \in \mathcal{A} = \mathbb{R}^{n+1} \label{eq: info operator 1st order}
\end{eqnarray}
which corresponds indirectly to $n+1$ evaluations of the original gradient field $f$ at certain input pairs $(x_i,y_i)$.
The selection of the inputs $r_i$ is discussed in Section \ref{subsec: training set}.

The transformation of a first order ODE is clearly illustrated in the following:
\begin{example}[Ex. \ref{1storderexample initial}, continued]\label{1storderexample}
Consider the first order ODE $\frac{\mathrm{d}\mathrm{y}}{\mathrm{d}x} = f(x,\mathrm{y}(x))$, $f(x,y) = F\left( \frac{y}{x} \right)$.
Recall from Ex. \ref{1storderexample initial} that this ODE admits the one parameter Lie group of transformations $x^* = \alpha x$, $y^* = \alpha y$ for $\alpha \in \R$ and the associated infinitesimal generator is $\mathrm{X}=x\frac{\partial}{\partial x}+y\frac{\partial}{\partial y}$. 
Solving the pair of partial differential equations $\mathrm{X}r=0, \mathrm{X}s=1$ yields the canonical coordinates $s = \log y$, $r = \frac{y}{x}$.
The transformed ODE is then $\frac{\mathrm{d}\mathrm{s}}{\mathrm{d}r} = \frac{F(r)}{-r^2+rF(r)} =: G(r)$.
Thus an evaluation $G(r)$ corresponds to an evaluation of $f(x,y)$ at an input $(x,y)$ such that $r = \frac{y}{x}$.
\end{example}

Two important points must now be addressed:
First, the approach just described cannot be Bayesian unless it corresponds to a well-defined prior distribution $\mu \in \mathcal{P}_{\mathcal{Y}}$ in the original coordinate system $\mathcal{Y}$.
This precludes standard (e.g. Gaussian process) priors in general, as such priors assign mass to functions in $(r,s)$-space that do not correspond to well-defined functions in $(x,y)$-space (see Fig. \ref{fig: not well defined}).
Second, any prior that is used ought to be consistent with the Lie group of transformations that the ODE is known to admit.
To address each of these important points, we propose two general principles for prior construction in this work.
The first principle is the \emph{implicit prior} principle.
This ensures that a prior specified in the transformed coordinates $(r,s)$ can be safely transformed into a well-defined distribution on $\mathcal{Y}$.
For such an implicit prior to be well-defined we need to understand when a function in $(r,s)$ space maps to a well-defined function in the original $(x,y)$ domain of interest.
Let $\mathcal{S}$ denote the image of $\mathcal{Y}$ under the canonical coordinate transformation.

\begin{figure}[t!]
\centering
\resizebox{.7\textwidth}{!}{
\begin{tikzpicture}
\node[anchor=south west,inner sep=0] (image) at (0,0) {
\includegraphics[width = 0.6\textwidth,clip,trim = 3.6cm 4cm 4cm 4cm,page=2]{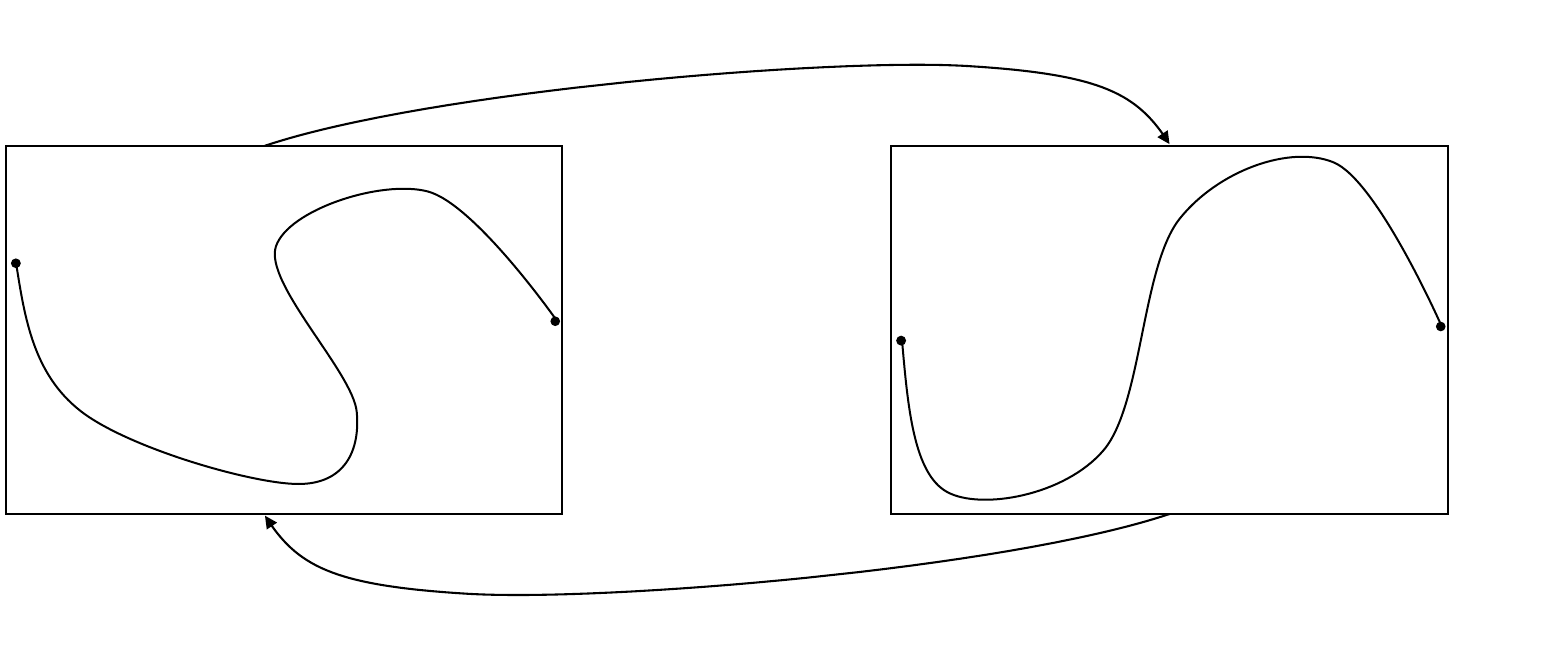}};
\node (A) at (-1.3,2.1) {$(x,y)$};
\node (B) at (10.5,2.1) {$(r,s)$};
\node (C) at (3.5,4) {Lie transform};
\node (D) at (4,0.05) {(Lie transform)$^{-1}$};
\end{tikzpicture} }
\caption{Illustration of the implicit prior principle: A prior elicited for the function $\mathrm{s}(r)$ in the transformed coordinate system $(r,s)$ must be supported on functions $\mathrm{s}(r)$ that correspond to well-defined functions $\mathrm{y}(x)$ in the original coordinate system $(x,y)$. Thus the situation depicted would not be allowed.}
\label{fig: not well defined}
\end{figure}
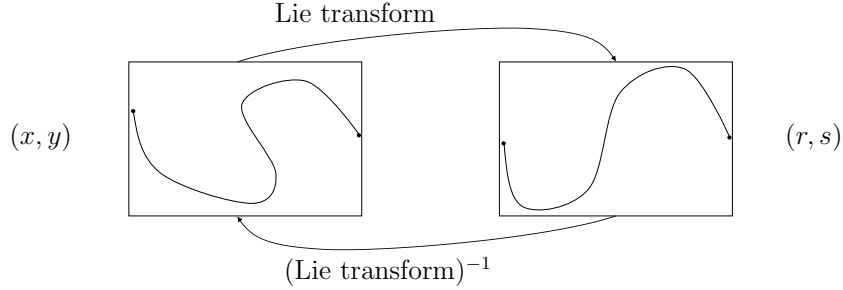

\begin{principle}[Implicit Prior]\label{principle: implicitprior}
A distribution $\nu \in \mathcal{P}_{\mathcal{S}}$ on the transformed solution space $\mathcal{S}$ corresponds to a well-defined \emph{implicit prior} $\mu \in \mathcal{P}_{\mathcal{Y}}$ provided that $x(r,\mathrm{s}(r))$ is strictly monotone as a function of $r$.
\end{principle}

\begin{example}[Ex. \ref{1storderexample}, continued] \label{1storderexample again 2}
For the ODE in Ex. \ref{1storderexample}, with canonical coordinates $s = \log y$, $r = \frac{y}{x}$, if $x \in [x_0,x_T] = [1,x_T]$ and $y \in (0,\infty)$, then the region in the $(r,s)$ plane corresponding to $[1,x_T] \times (0,\infty)$ in the $(x,y)$ plane is $(0,\infty) \times \R$.
Now, 
\begin{eqnarray*}
\frac{\mathrm{d}x(r,\mathrm{s}(r))}{\mathrm{d}r} \; = \; \frac{\partial{x}}{\partial r}+\frac{\partial{x}}{\partial s}\frac{\mathrm{d} \mathrm{s}(r)}{\mathrm{d}r} \; = \; \frac{r \mathrm{s}'(r) \exp(\mathrm{s}(r))-\exp(\mathrm{s}(r))}{r^2} .
\end{eqnarray*}
Thus $\frac{\mathrm{d}x}{\mathrm{d}r} > 0$ if and only if $\mathrm{s}'(r) > \frac{1}{r}$ and the invariant prior principle requires that we respect the constraint $\log(r) \le \mathrm{s}(r) \le \log(r)+\log(x_T)$ for all $r>0$.
The set $\mathcal{S}$ must therefore consists of differentiable functions $\mathrm{s}$ defined on $r \in (0,\infty)$ and satisfying $\log (r) \leq \mathrm{s}(r) \leq \log(r) + \log(x_T)$.
\end{example}

Now we turn to the second important point, namely that the prior ought to encode knowledge about any Lie transformations that are known to be admitted by the ODE.
In working on the transformed space $\mathcal{S}$, it become clear how to construct a prior measure in which this knowledge is encoded.
Our second principle for prior specification states that equal prior weight should be afforded to all curves that are identical up to a Lie transformation:

\begin{principle}[Invariant Prior]
A distribution $\nu \in \mathcal{P}_{\mathcal{S}}$ on the transformed solution space $\mathcal{S}$ is said to be \emph{invariant} provided that $\nu(S) = \nu(S + \epsilon)$ where the elements of $S + \epsilon$ are the elements of $S$ after a vertical translation; i.e. $\mathrm{s}(\cdot) \mapsto \mathrm{s}(\cdot) + \epsilon$ and both $S, S + \epsilon \in \Sigma_{\mathcal{S}}$.
\end{principle}

Our recommendation is that, when possible, both the implicit prior principle and the invariant prior principle should be enforced.
However, in practice it seems difficult to satisfy both principles and our empirical results in Section \ref{sec: experiment} are based on implicit priors that are not invariant.

\subsection{The Case of a Second Order ODE}
\label{subsec: prior 2nd order}

In this section we present our approach for a second order ODE.
The study of second order ODEs is particularly important, since Newtonian mechanics is based on ODEs of second order.
The presentation is again simplified relative to the general case of an $m$th order ODE, this time by virtue of the fact that any two dimensional Lie algebra is guaranteed to be solveable (Thm. \ref{thm: 2dsolvable}).
The main result that will allow us to construct an exact Bayesian PNM is as follows:

\begin{theorem}[Reduction of a Second Order ODE to Two Integrals]\label{2ndordertoquad}
If a second order ODE 
\begin{eqnarray}\label{2ndorder}
\frac{\mathrm{d}^2\mathrm{y}}{\mathrm{d}x^2} & = & f\left(x, \mathrm{y}(x) , \frac{\mathrm{d}\mathrm{y}}{\mathrm{d}x} \right)
\end{eqnarray}
admits a two parameter Lie group of transformations, then there exists coordinates $r(x,y)$, $s(x,y)$ such that 
\begin{eqnarray}
\frac{\mathrm{d}\mathrm{s}}{\mathrm{d}r} & = & G(r) \label{eq: thm conc 2nd order}
\end{eqnarray}
for some implicitly defined function $G$.
The function $G$ is explicitly related to the solution of a second equation of the form
\begin{eqnarray}
\frac{\mathrm{d}\tilde{\mathrm{s}}}{\mathrm{d}\tilde{r}} & = & H(\tilde{r})  \label{eq: thm conc 2nd order 2}
\end{eqnarray}
for some explicit function $H(\tilde{r})$.
\end{theorem}

Note that the ODE in Eq.~\eqref{2ndorder} is reduced to two integrals, namely Eq.~\eqref{eq: thm conc 2nd order} and Eq.~\eqref{eq: thm conc 2nd order 2}.
At a high level, our proposed Bayesian PNM performs inference for the solution $\mathrm{s}(r)$ of Eq.~\eqref{eq: thm conc 2nd order} and then transforms the resultant posterior back into the original $(x,y)$-coordinate system.
However, because $G$ in Eq.~\eqref{eq: thm conc 2nd order} depends on the solution $\tilde{\mathrm{s}}(\tilde{r})$ of Eq.~\eqref{eq: thm conc 2nd order 2}, we must also estimate $\tilde{\mathrm{s}}(\tilde{r})$ and for this we need to evaluate $H$.
Our PNM is therefore based on the information operator
\begin{eqnarray*}
A(\mathrm{y}) & = & \left[ G(r_0) , \dots , G(r_n) , H(\tilde{r}_0) , \dots , H(\tilde{r}_n) \right]^\top \in \mathcal{A} = \mathbb{R}^{2(n+1)}
\end{eqnarray*}
which corresponds indirectly to $2(n+1)$ evaluations of $f$, the original gradient field.
The extension of our approach to a general $m$th order ODE proceeds analogously, with $\mathcal{A} = \mathbb{R}^{m(n+1)}$.
The use of Thm. \ref{2ndordertoquad} is illustrated in Example \ref{ex: 2nd ord reduce} in the Supplement.

The two principles of prior construction that we advocated in the case of a first order ODE apply equally to the case of a second- and higher-order ODE.
It therefore remains only to discuss the selection of the specific inputs $r_i$ (and $\tilde{r}_i$ in the case of a second order ODE) that are used to define the information operator $A$.
This discussion is again reserved for Supplemental Section \ref{subsec: training set}.

\section{Numerical Illustration} \label{sec: experiment}

In this section the proposed Bayesian PNM is empirically illustrated.
Recall that we are not advocating these methods for practical use, rather they are to serve as a proof-of-concept for demonstrating that exact Bayesian inference can \emph{in principle} be performed for ODEs, albeit at considerable effort; a non-trivial finding that helps to shape ongoing research and discussion in this nascent field.

The case of a first order ODE is considered in Section \ref{subsec: 1st order experiment} and the second order case is contained in Section \ref{subsec: 2nd order experiment}.
In both cases, scope is limited to verifying the correctness of the procedures, as well as indicating how conjugate prior distributions can be constructed.

\subsection{A First Order ODE} \label{subsec: 1st order experiment}

This section illustrates the empirical performance of the proposed method for a first order ODE.

\paragraph{ODE:}
To limit scope we consider first order ODEs of the form
\begin{equation}
\frac{\mathrm{d}\mathrm{y}}{\mathrm{d}x} = F\left( \frac{\mathrm{y}(x)}{x} \right) , \qquad x \in [1,x_T], \qquad \mathrm{y}(1) = y_0 . \label{eq: experiment 1st order}
\end{equation}
Note that admitted transformation and associated canonical coordinates for this class of ODE have already been derived in Ex. \ref{1storderexample initial}, Ex. \ref{1storderexample} and Ex. \ref{1storderexample again 2}.

\paragraph{Prior:}
In constructing a prior $\mu \in \mathcal{P}_{\mathcal{Y}}$ we refer to the implicit prior principle in Sec. \ref{subsec: prior 1st order}.
Indeed, recall from Ex. \ref{1storderexample} that the ODE in Eq.~\eqref{eq: experiment 1st order} can be transformed into an ODE of the form
\begin{eqnarray*}
\frac{\mathrm{d} \mathrm{s}}{\mathrm{d} r} = G(r), \qquad r \in (0,\infty) , \qquad \mathrm{s}(y_0) = \log(y_0) .
\end{eqnarray*}
Then our approach constructs a distribution $\nu \in \mathcal{P}_{\mathcal{S}}$ where, from Ex. \ref{1storderexample again 2}, $\mathcal{S}$ is the set of differentiable functions $\mathrm{s}$ defined on $r \in (0,\infty)$ and satisfying 
\begin{eqnarray}
\log (r) \; \leq \; \mathrm{s}(r) \; \leq \; \log(r) + \log(x_T). \label{eq: domain constraint}
\end{eqnarray}
To ensure monotonicity in the implicit prior principle, we take $\frac{\mathrm{d}x}{\mathrm{dr}} > 0$, which translates into the requirement that
\begin{eqnarray}
\frac{\mathrm{d}\mathrm{s}}{\mathrm{d}r} & > & \frac{1}{r} . \label{eq: implicit prior 1st order}
\end{eqnarray}
If Eq.~\eqref{eq: implicit prior 1st order} holds, then $\nu$ induces a well-defined distribution $\mu \in \mathcal{P}_{\mathcal{Y}}$.
Note that the constraints in Eq.~\eqref{eq: domain constraint} and Eq.~\eqref{eq: implicit prior 1st order} preclude the direct use of standard prior models, such as Gaussian processes.
However, it is nevertheless possible to design priors that are convenient for a given set of canonical coordinates.
Indeed, for the canonical coordinates $r,s$ in our example, we can consider a prior of the form
\begin{eqnarray}
\mathrm{s}(r) & = & \log(r) + \log(x_T)\zeta(r) \nonumber
\end{eqnarray}
where the function $\zeta : (0,\infty) \rightarrow \mathbb{R}$ satisfies
\begin{eqnarray}
\zeta(y_0)=0, \qquad \zeta(r) \leq 1, \qquad \frac{\mathrm{d}\zeta}{\mathrm{d}r} \geq 0  .
\label{eq: prior criterion 1st order}
\end{eqnarray}
For this experiment, the approach of \cite{Lopez-Lopera2017} was used as a prior model for the monotone, bounded function $\zeta$; this requires that a number, $N$, of basis functions is specified - for brevity we defer the detail to Appendix \ref{subsec: Lopez details}.

\begin{figure}[t!]
\centering
\includegraphics[align=c, width = 0.43\textwidth,clip,trim = 5cm 9.5cm 5cm 9cm]{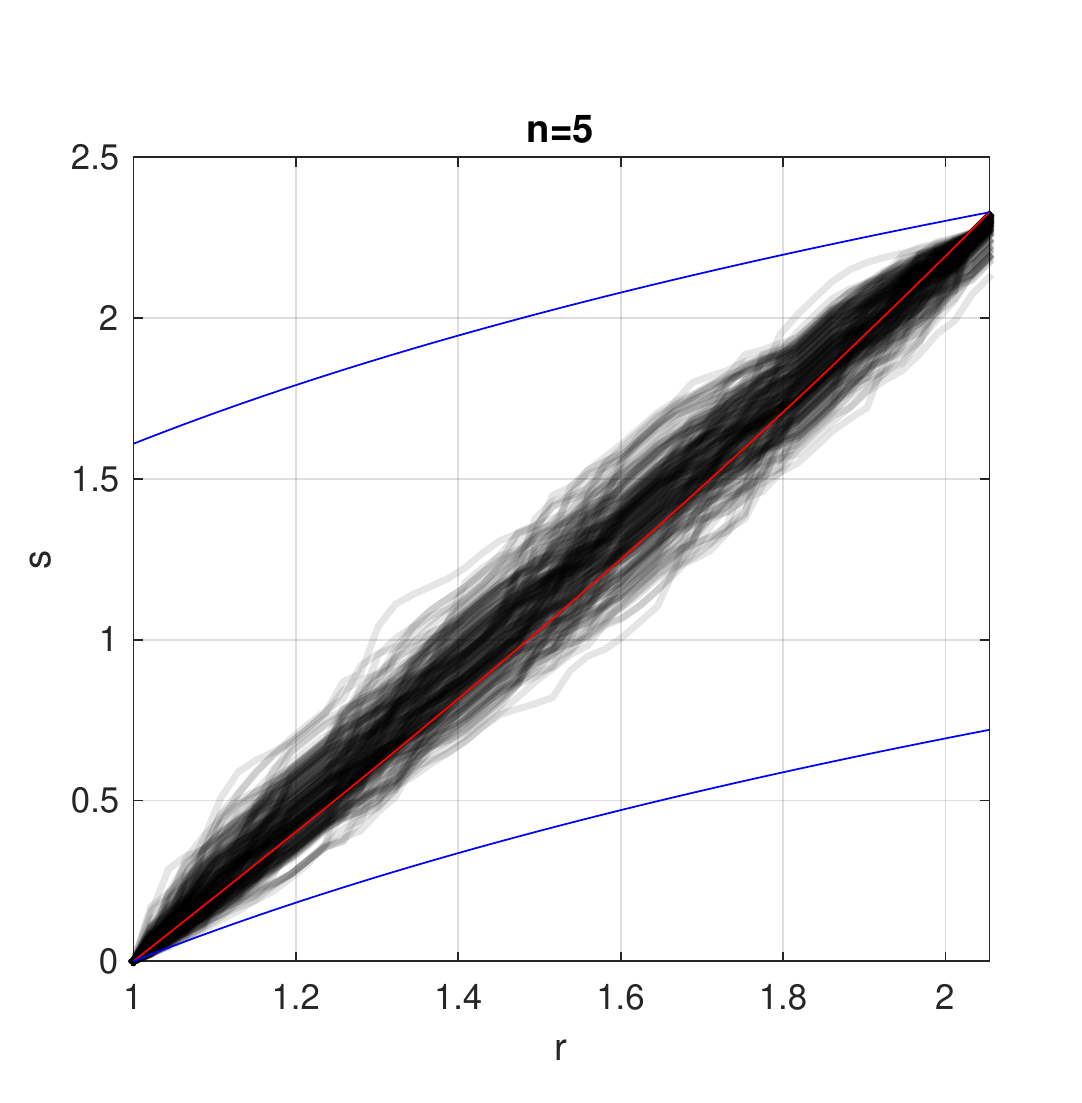}
\includegraphics[align=c, width = 0.42\textwidth,clip,trim = 5cm 8.3cm 5cm 8cm]{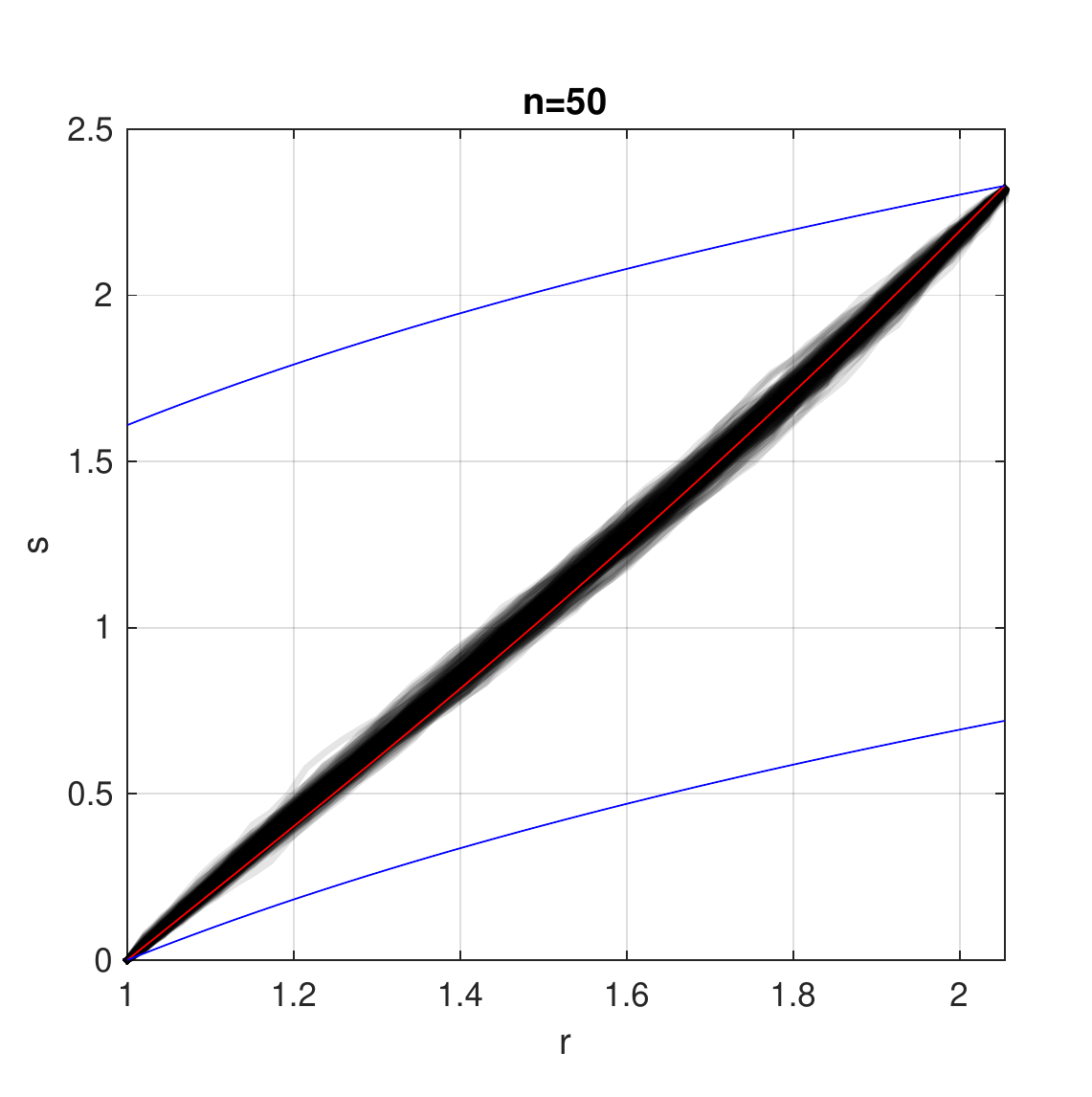}

\includegraphics[align=c, width = 0.43\textwidth,clip,trim = 5cm 9.3cm 5cm 9.5cm]{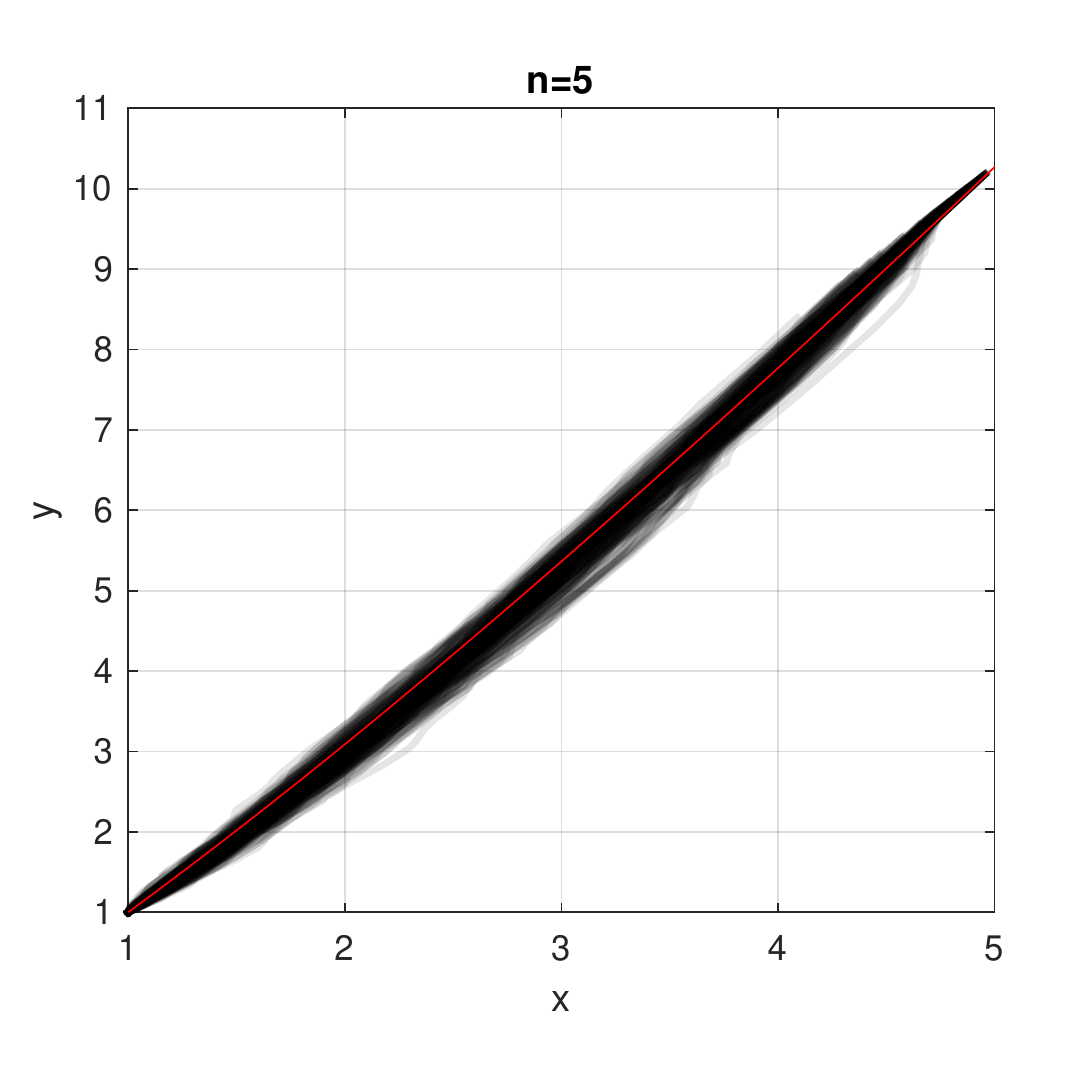}
\includegraphics[align=c, width = 0.41\textwidth,clip,trim = 5cm 7.9cm 5cm 8cm]{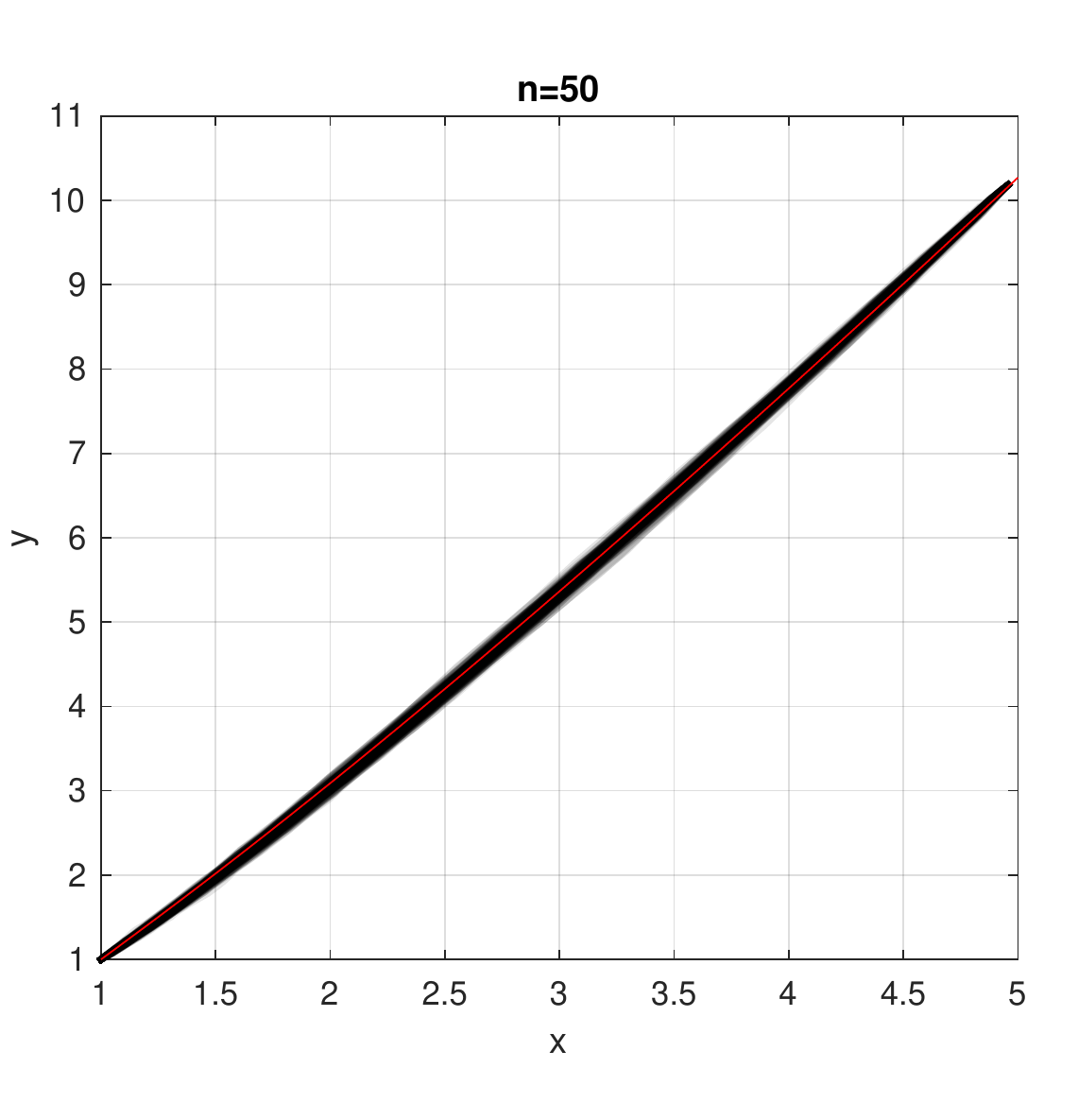}
\caption{
Experimental results, first order ODE: 
The black curves represent samples from the posterior, whilst the exact solution is indicated in red.
The blue curves represent a constraint on the $(r,s)$ domain that arises when the implicit prior principle is applied.
The number $n$ of gradient evaluations is indicated.
Top: results in the $(r,s)$ domain.
Bottom: results in the $(x,y)$ domain.}
\label{fig: 1st order results rs}
\end{figure}

The prior just described incorporates the symmetric structure of the ODE, in the sense that the independent variable $r=\frac{y}{x}$ is the first canonical coordinate of the infinitesimal generator of the Lie group of transformations of the original ODE in Eq. \ref{1storderexample}. 
In other words, $r$ is a variable fixed by the Lie group of transformations of the ODE (in this case $x^* = \alpha x$, $y^* = \alpha y$, so $r^*=r$). 
Because the prior is defined on functions $s(r)$ of $r$, this means the prior itself is also unchanged by the Lie group of transformations of the ODE, so that the prior effectively incorporates the symmetric structure of the ODE. 

\paragraph{Results:}
To obtain empirical results we consider the ODE with $F(r) = r^{-1} + r$ and $y_0 = 1, x_T=5$.
The posterior distributions that were obtained as the number $n$ of data points was increased were sampled and plotted in the $(r,s)$ and $(x,y)$ planes in Fig. \ref{fig: 1st order results rs}.
In each case a basis of size $N = 2n$ was used.
Observe that the implicit prior principle ensures that all curves in the $(x,y)$ plane are well-defined functions (i.e. there is at most one $y$ value for each $x$ value).
Observe also that the posterior mass appears to contract to the true solution $\mathrm{y}^\dagger$ of the ODE as the number of evaluations $n$ of the gradient field is increased.

\subsection{A Second Order ODE} \label{subsec: 2nd order experiment}

This section illustrates the empirical performance of the proposed method for a second order ODE.

\paragraph{ODE:}
Consider again the second order nonlinear ODE in Eqn. \ref{eq:2ndorderexample} together with the initial condition $y(x_0) = y_0$, $\frac{\mathrm{d}\mathrm{y}}{\mathrm{d}x}(x_0) = y_0'$.

\paragraph{Prior:}
It is shown in Ex. \ref{ex: 2nd ord reduce} in the Supplement that Eq.~\ref{eq:2ndorderexample} can be reduced to a first order ODE in $(s,r)$ with $-\frac{1}{x_0} - r \leq s \leq -\frac{1}{x_T}-r$.
The implicit prior principle in this case requires that $\frac{\mathrm{d}\mathrm{s}}{\mathrm{d}r} > -1$.
Thus we are led to consider a parametrisation of the form
\begin{eqnarray}
\mathrm{s}(r) & = & - \frac{1}{x_0} -r + \left(\frac{1}{x_0} - \frac{1}{x_T} \right) \zeta(r) \nonumber
\end{eqnarray}
where the function $\zeta$ again satisfies the conditions in Eq.~\eqref{eq: prior criterion 1st order}.
The approach of \cite{Lopez-Lopera2017} was therefore again used as a prior model.

For this example an additional level of analytic tractability is possible, as described in detail in Ex. \ref{ex: 2nd ord reduce} in the Supplement.
Thus we need only consider an information operator of the form $A(\mathrm{y}) = [G(r_0) , \dots , G(r_n)]$.

\paragraph{Results:}
The posterior distributions that were obtained are plotted in the $(r,s)$ plane and the $(x,y)$ plane in Fig. \ref{fig: 2nd order results}.
A basis of size $N = 2n$ was used, with $[y_0,y'_0] = [-10,1], [x_0,x_T]=[5,10]$.
Observe that the implicit prior principle ensures that all curves in the $(x,y)$ plane are well-defined functions (i.e. there is at most one $y$ value for each $x$ value). \textcolor{black}{The true solution appears to be smoother than the samples, even for 50 gradient evaluations, which suggests that the prior was somewhat conservative in this context.}

\begin{figure}[t!]
\centering
\includegraphics[width = 0.48\textwidth,clip,trim = 4cm 9cm 5cm 9cm]{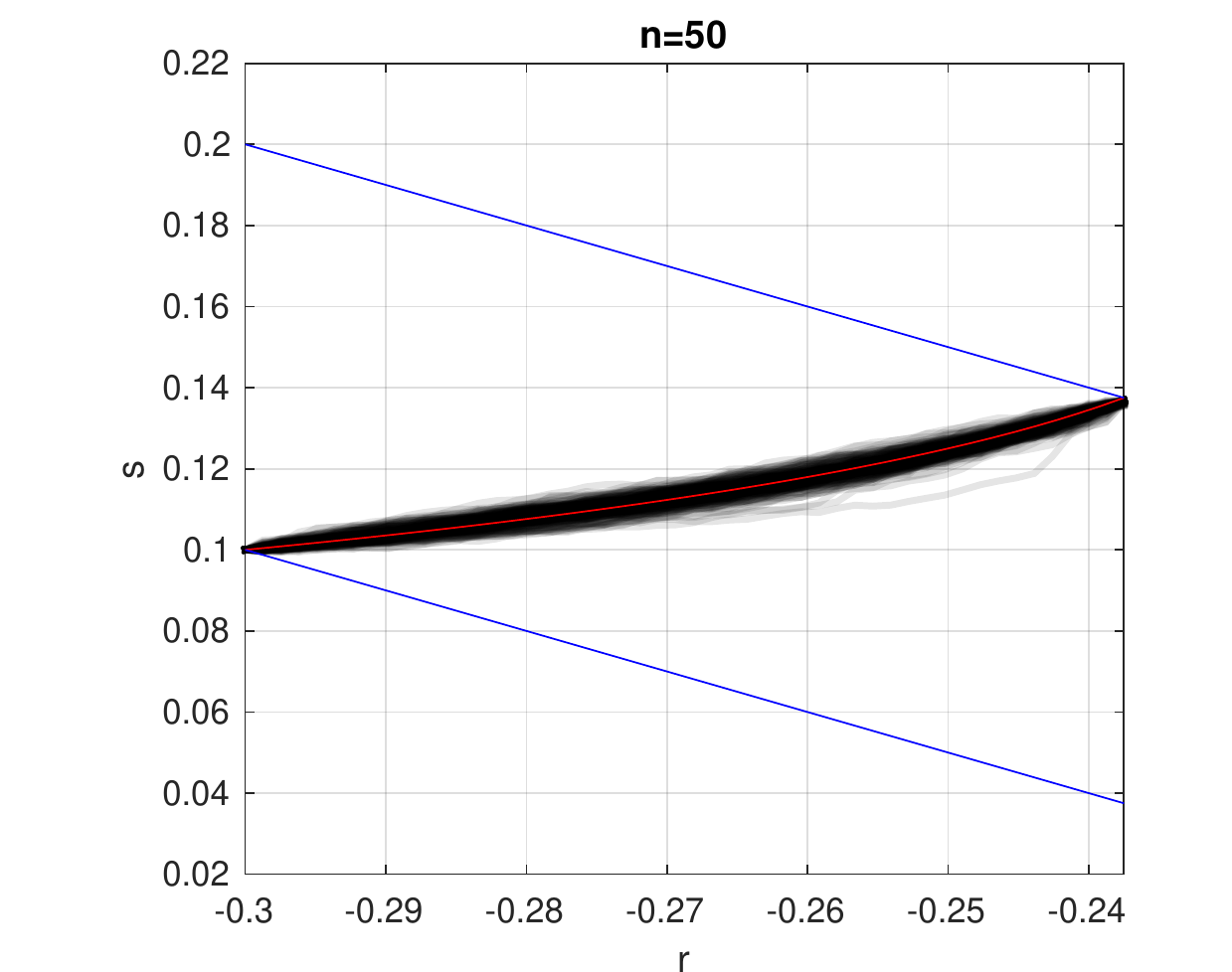}
\includegraphics[width = 0.48\textwidth,clip,trim = 4cm 9cm 5cm 9cm]{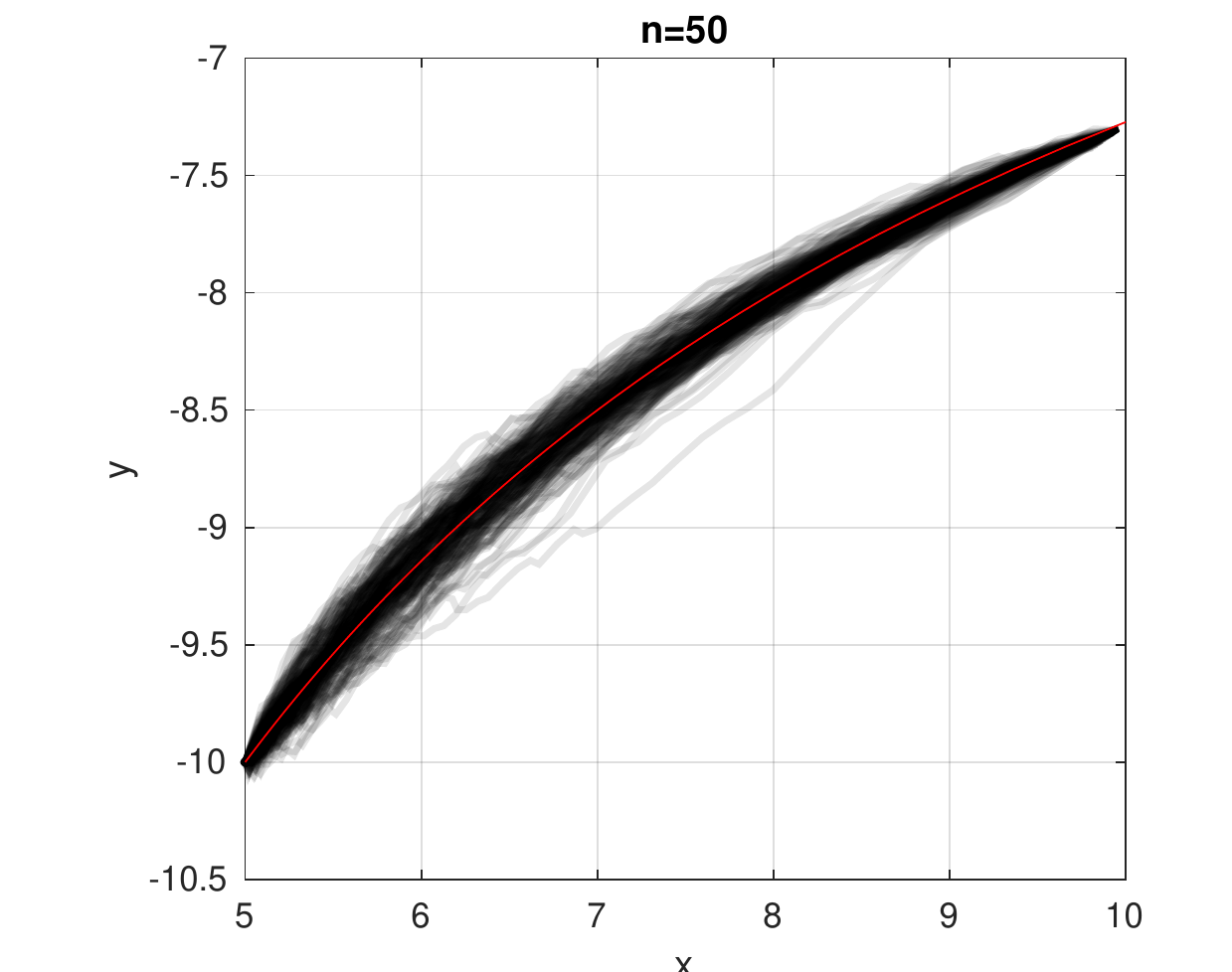}
\caption{
Experimental results, second order ODE: 
The black curves represent samples from the posterior in the $(r,s)$ plane (left) and $(x,y)$ plane (right), whilst the exact solution is indicated in red.
The blue curves represent a constraint on the domain that arises when the implicit prior principle is applied.
The number of gradient evaluations was $n = 50$.
}
\label{fig: 2nd order results}
\end{figure}

\section{Conclusion} \label{sec: conclusion}

This paper presented a foundational perspective on PNM.
It was first argued that there did not exist a Bayesian PNM for the numerical solution of ODEs.
Then, to address this gap, a prototypical Bayesian PNM was developed.
The Bayesian perspective that we have put forward sheds light on foundational issues which will need to be addressed going forward:

\paragraph{Foundation of PNM:}

As explained in Section \ref{subsec: existing work}, existing PNM for ODEs each take the underlying state space $\mathcal{Y}$ to be the solution space of the ODE.
This appears to be problematic, in the sense that a generic evaluation $f(x_i,y_i)$ of the gradient field cannot be cast as information $A(\mathrm{y}^\dagger)$ about the solution $\mathrm{y}^\dagger$ of the ODE unless the point $(x_i,y_i)$ lies exactly on the solution curve $\{(x,\mathrm{y}^\dagger(x)) : x \in [x_0,x_T]\}$.
As a consequence, all existing PNM of which we are aware violate the likelihood principle and are therefore not strictly Bayesian.
The assumption of a solvable Lie algebra, used in this work, can be seen as a mechanism to ensure the existence of an exact information operator $A$, so that the likelihood principle can be obeyed.
However, for a general ODE it might be more natural to take the underlying state space to be a set $\mathcal{F}$ of permitted gradient fields and the quantity of interest $Q(f)$ to map a gradient field $f$ to the solution of the associated ODE.
This would make the information operator $A$ trivial but evaluation of the push-forward $Q_{\#} \mu^a$ would require the exact solution operator of the ODE.
However, the reliance on access to an oracle solver $Q$ makes this philosophically somewhat distinct from PNM.

\paragraph{Limitations of Bayesian PNM:}

The proposed method was intended as a proof-of-concept and it is therefore useful to highlight the aspects in which it is limited.
First, when an $m$th order ODE admits an $r$-parameter Lie group of transformations with $r > m$, there is an arbitrariness to the particular $m$-dimensional sub-group of transformations that are selected.
Second, the route to obtain transformations admitted by the ODE demands that some aspects of the gradient field $f$ are known, in contrast to other work in which $f$ is treated as a black-box.
For instance, in Ex. \ref{1storderexample} we used the fact that $f$ can be expressed as $f(x,y) = F(\frac{y}{x})$, although knowledge of the form of $F$ was not required.
Third, the class of ODEs for which a solvable Lie algebra is admitted is relatively small.
On the other hand, references such as \cite{Bluman2002} document important cases where our method could be applied.
Fourth, the principles for prior construction that we identified do not entail a unique prior and, as such, the question of prior elicitation must still be addressed. 

\paragraph{Outlook:}

The goal of providing rigorous and exact statistical uncertainty quantification for the solution of an ODE is, we believe, important and will continue to be addressed.
Traditional numerical methods have benefitted from a century of research effort and, in comparison, Bayesian PNM is an under-developed field. 
For example, the limited existing work on PNM for ODEs, such as \cite{Skilling1992,Schober2014,Chkrebtii2013,Kersting2016,Schober2016,Kersting2018,Tronarp2018},
does not attempt to provide adaptive error control \citep[though we note promising ongoing research in that direction by][]{Chkrebtii2019}. 
\textcolor{black}{Nevertheless, the case for developing Bayesian numerical methods - which shares some parallels with the case for Bayesian statistics as opposed to other inferential paradigms - is clear, as argued in \cite{Diaconis1988} and \cite{Hennig2015}. }
The insights we have provided in this paper serve to highlight the foundational issues pertinent to Bayesian PNM for ODEs.
Indeed, our proof-of-concept highlights that performing exact Bayesian inference for ODEs may be extremely difficult.
This in turn provides motivation for the continued development of `approximately Bayesian' approaches to PNM, which in Sec. \ref{subsec: existing work} we surveyed in detail.

\paragraph{Acknowledgements:}

The authors are grateful to Mark Craddock, Fran\c{c}ois-Xavier Briol and Tim Sullivan for discussion of this work, as well as to an Associate Editor and two Reviewers for their challenging but constructive feedback.
JW was supported by the EPSRC Centre for Doctoral Training in Cloud Computing for Big Data at Newcastle University, UK.
CJO was supported by the Lloyd's Register Foundation programme on data-centric engineering at the Alan Turing Institute, UK.
This material was based upon work partially supported by the National Science Foundation under Grant DMS-1127914 to the Statistical and Applied Mathematical Sciences Institute. Any opinions, findings, and conclusions or recommendations expressed in this material are those of the authors and do not necessarily reflect the views of the National Science Foundation.

\newpage
\appendix
\setcounter{page}{1}
\section{Supplementary material}

This electronic supplement to the paper by Wang, Cockayne and Oates contains proofs for theoretical results in the main text (Sec. \ref{proof section}), a brief discussion on Bayesian experimental design in the context of designing a Bayesian probabilistic numerical method (Sec. \ref{subsec: training set}), and computation detail that was suppressed from the main text (Sec. \ref{subsec: Lopez details}).

\subsection{Proof of Theoretical Results} \label{proof section}

\begin{proof}[Proof of Theorem \ref{thm: exp generator}]
From Taylor's theorem we have that
\begin{eqnarray*}
x^* & = & X(x,\epsilon) \\
& = & \sum_{k=0}^\infty \frac{\epsilon^k}{k!} \left. \frac{\partial^k X(x,\epsilon)}{\partial \epsilon^k} \right|_{\epsilon = 0}
\end{eqnarray*}
For any differentiable function $F$ we have that
\begin{eqnarray*}
\frac{\mathrm{d}F(x^*)}{\mathrm{d}\epsilon} \; = \; \sum_{i=1}^d \frac{\partial F(x^*)}{\partial x_i^*} \frac{\mathrm{d}x_i^*}{\mathrm{d}\epsilon} \; = \; \sum_{i=1}^d \xi_i \frac{\partial F(x^*)}{\partial x_i^*} \; = \; \mathrm{X} F(x^*)
\end{eqnarray*}
and similarly
\begin{eqnarray*}
\frac{\mathrm{d}^kF(x^*)}{\mathrm{d}\epsilon^k} & = & \mathrm{X}^k F(x^*).
\end{eqnarray*}
Thus
\begin{eqnarray*}
\left. \frac{\partial^k X(x,\epsilon)}{\partial \epsilon^k} \right|_{\epsilon = 0} & = & \mathrm{X}^k x
\end{eqnarray*}
so that the stated result is recovered.
\end{proof}

\begin{proof}[Proof of Theorem \ref{thm: invar 1}]
The result is established as follows:
\begin{eqnarray*}
F \text{ invariant} & \Leftrightarrow & F(x^*) = 0 \text{ whenever } F(x) = 0 \\
& \Leftrightarrow & e^{\epsilon \mathrm{X}}F(x) = 0 \text{ whenever } F(x) = 0 \quad \text{(Cor. \ref{cor: exp})} \\
& \Leftrightarrow & F(x) + \epsilon \mathrm{X} F(x) + O(\epsilon^2) = 0 \text{ whenever } F(x) = 0 \quad \text{(Taylor)} \\
& \Leftrightarrow & \mathrm{X} F(x) = 0 \text{ whenever } F(x) = 0
\end{eqnarray*} 
where the last line follows since the coefficient of the $O(\epsilon)$ term in the Taylor expansion must vanish.
This completes the proof.
\end{proof}

\begin{proof}[Proof of Theorem \ref{thm: invar 2}]
From Cor. \ref{cor: exp}, we have that $F(x^*) = F(x) + \epsilon \mathrm{X} F(x) + O(\epsilon^2)$.
The result follows from inspection of the $\epsilon$ coefficient.
\end{proof}

\begin{proof} [Proof of Theorem \ref{thm: 2dsolvable}] 
Suppose $\mathcal{L}$ is a two dimensional Lie Algebra generated by linearly independent infinitesimal generators $\mathrm{X}_1$ and $\mathrm{X}_2$. Let $\mathrm{Y}=[\mathrm{X}_1,\mathrm{X}_2]=a\mathrm{X}_1+b\mathrm{X}_2$ and let $\mathcal{J}$ be the one dimensional subalgebra generated by $\mathrm{Y}$. Suppose $\mathrm{Z}=c\mathrm{X}_1+d\mathrm{X}_2$ is some element of $\mathcal{L}$, then 
\begin{eqnarray*}
[\mathrm{Y}, \mathrm{Z}]& = & \;[Y,c\mathrm{X}_1+d\mathrm{X}_2]\\
& = & \; c[\mathrm{Y},\mathrm{X}_1]+d[\mathrm{Y},\mathrm{X}_2]\\
& = & \; cb[\mathrm{X}_2,\mathrm{X}_1]+da[\mathrm{X}_1,\mathrm{X}_2]\\
& = & \; (ad-bc)\mathrm{Y} \in \mathcal{J}
\end{eqnarray*} 
So $\mathcal{J} \subset \mathcal{L}$ is normal, thus $\mathcal{L}$ is solvable, as claimed.
\end{proof}

\begin{proof} [Proof of Theorem \ref{1stordertoquad}]
Let the infinitesimal generator associated with the Lie group of transformations be denoted $\mathrm{X}$.
From the remarks after Def. \ref{def: canonical}, we can obtain canonical coordinates by solving the pair of first order partial differential equations $\mathrm{X}r=0$, $\mathrm{X}s=1$. 
By the chain rule we have 
\begin{eqnarray*}
\frac{\mathrm{d}s}{\mathrm{d}r} & = & \frac{s_x+s_{y}y'}{r_x+r_{y}y'} \; =: \; G(r,s)
\end{eqnarray*}
From the definition of canonical coordinates, the Lie group of transformations is $r^*=r$, $s^*=s+\epsilon$ in the transformed coordinate system, so 
\begin{eqnarray*}
\cfrac{\mathrm{d}s^*}{\mathrm{d}r^*} = G(r^*,s^*) & \implies & \cfrac{\mathrm{d}s}{\mathrm{d}r}=G(r,s+\epsilon)
\end{eqnarray*}
for all $\epsilon$, which implies $G(r,s)=G(r)$ and thus Eq.~\eqref{1storder} becomes
\begin{eqnarray*}
\cfrac{\mathrm{d}\mathrm{s}}{\mathrm{d}r} & = & G(r)
\end{eqnarray*}
as claimed.
\end{proof}

\begin{proof}[Proof of Theorem \ref{2ndordertoquad}]

Let the infinitesimal generators associated with the Lie group of transformations be denoted $\mathrm{X}_1$ and $\mathrm{X}_2$.
Recall from Thm. \ref{thm: 2dsolvable} that any two dimensional Lie algebra is solvable.
Thus, without loss of generality we may assume 
\begin{eqnarray} \label{eq: 2nd order commutator}
[\mathrm{X_1},\mathrm{X_2}] & = & \lambda\mathrm{X_1}
\end{eqnarray}
for some $\lambda \in \mathbb{R}$.
The infinitesimal generators $\mathrm{X}_1$ and $\mathrm{X}_2$ each correspond to a one parameter Lie group of transformations, denoted $x^*=X_1(x,\epsilon_1)$ and $x^\dagger=X_2(x,\epsilon_2)$.
Let $v(x,y)$, $w(x,y,y_1)$ be invariant functions of $\mathrm{X_1}$ and its extension $\mathrm{X_1^{(1)}}$, respectively, so $v(x^*,y^*)=v(x,y)$ and $w(x^*,y^*,y_1^*)=w(x,y,y_1)$, where $w$ has a non-trivial dependence on its third argument.
It follows from the definition of invariance that
\begin{eqnarray*}
\frac{\mathrm{d}w^*}{\mathrm{d}v^*} & = & \frac{\mathrm{d}w}{\mathrm{d}v}, 
\end{eqnarray*}
which is equivalent to 
\begin{eqnarray}
\mathrm{X}_1^{(1)}\cfrac{\mathrm{d}w}{\mathrm{d}v} & = & 0  \label{eq: 2nd order pf 1}
\end{eqnarray}
by Cor. \ref{thm: invar 3}. Now because Eq.~\eqref{eq: 2nd order pf 1} is a homogeneous partial differential equation, the general solution $\frac{\mathrm{d}w}{\mathrm{d}v}$ can be expressed as a function of the two solutions $v(x,y)$ and $w(x,y,y_1)$.
Therefore
\begin{eqnarray}\label{eq:1streduction}
\cfrac{\mathrm{d}w}{\mathrm{d}v} & = & Z(v,w)
\end{eqnarray}
for some undetermined function $Z$.

Since Eq.~\eqref{2ndorder} admits $\mathrm{X_2}$, and Eq.~\eqref{eq:1streduction} is the same ODE when expressed in terms of $x,y,y_1$, it must be the case that 
\begin{eqnarray*}
\mathrm{X_2^{(2)}} \left( \frac{\mathrm{d}w}{\mathrm{d}v}-Z(v,w) \right) \; = \; 0 \qquad \text{whenever} \qquad \frac{\mathrm{d}w}{\mathrm{d}v}\; = \; Z(v,w). 
\end{eqnarray*}
Then from Cor. \ref{thm: Invar criteria} it follows that $\mathrm{X}_2^{(1)}$ is admitted by the first order ODE in Eq.~\eqref{eq:1streduction}.
Thus we are now faced with a first order ODE that admits a one parameter Lie group of transformations, as in Thm. \ref{1stordertoquad}.

Now, from Eq.~\eqref{eq: 2nd order commutator}, we have $\mathrm{X_1}\mathrm{X_2}v=\mathrm{X_2}\mathrm{X_1}v+\lambda\mathrm{X_1}v=0$. 
Thus $\mathrm{X}_2 v$ is an invariant of $\mathrm{X}_1$ and so $\mathrm{X_2}v=h(v)$ for some function $h$. 
Similarly $\mathrm{X_1^{(1)}}\mathrm{X_2^{(1)}}v=\mathrm{X_2^{(1)}}\mathrm{X_1^{(1)}}v+\lambda\mathrm{X_1^{(1)}}v=0$, so that $\mathrm{X_2^{(1)}}w=g(v,w)$, for some function $g$.
This implies $\mathrm{X}_2^{(1)}=h(v)\frac{\partial}{\partial v}+g(v,w)\frac{\partial}{\partial w}$.

Proceeding as in Thm. \ref{1stordertoquad}, denote the canonical coordinates of $\mathrm{X_2^{(1)}}=h(v)\frac{\partial}{\partial v}+g(v,w)\frac{\partial}{\partial w}$ by $\tilde{r}(v,w)$, $\tilde{s}(v,w)$ such that $\mathrm{X_2^{(1)}}\tilde{r}=0$, $\mathrm{X_2^{(1)}}\tilde{s}=1$. 
In canonical coordinates, Eq.~\eqref{eq:1streduction} becomes: 
\begin{eqnarray}
\frac{\mathrm{d}\tilde{\mathrm{s}}}{\mathrm{d}\tilde{r}} & = & H(\tilde{r}) \label{eq: tilde ODE eq}
\end{eqnarray}
This is again an integral, with solution
\begin{eqnarray}
\tilde{\mathrm{s}}(\tilde{r}) & = & \int^{\tilde{r}} H(t) \mathrm{d}t+C . \label{eq:2ndreduction}
\end{eqnarray}
We can rewrite Eq.~\eqref{eq:2ndreduction} in terms of $v,w$ to obtain an equation of the form
\begin{eqnarray}\label{eq:transformback}
I(v,w) & = & 0 
\end{eqnarray}
which satisfies $\mathrm{X_1^{(1)}}(I(v,w))=0$ whenever $I(v,w)=0$, since recall $v, w$ are invariants of $X_1^{(1)}$. 
For the final step, we recall that $v = v(x,y)$ and $w = w(x,y,y_1)$, so that Eq.~\eqref{eq:transformback} represents a first order ODE in $y$, which admits $\mathrm{X_1}$.
Thus we can apply Thm. \ref{1storder} a second time to obtain canonical coordinates $r(x,y)$, $s(x,y)$ for $X_1$.
In these coordinates, Eq.~\eqref{eq:transformback} reduces into the form 
\begin{eqnarray*}
\frac{\mathrm{d}\mathrm{s}}{\mathrm{d}r} & = & G(r)
\end{eqnarray*}
where $G$ is implicitly defined.
\end{proof}

\begin{example}[Deriving the Infinitesimal Generators for the Second Order ODE in Eq.~\ref{2ndorderexample0}] \label{ex: 2nd nonlinear}
Consider the second order nonlinear ODE
\begin{equation} \label{eq:2ndorderexample}
(x-\mathrm{y}(x))\frac{\mathrm{d}^2\mathrm{y}}{\mathrm{d}x^2}+2\frac{\mathrm{d}\mathrm{y}}{\mathrm{d}x} \left( \frac{\mathrm{d}\mathrm{y}}{\mathrm{d}x}+1 \right) +\left(\frac{\mathrm{d}\mathrm{y}}{\mathrm{d}x} \right)^{3/2} = 0 .
\end{equation}
Using Corollary  \ref{thm: Invar criteria}, we have:
$$
\left(\xi\frac{\partial}{\partial x}+\eta\frac{\partial}{\partial y}+\eta^{(1)}\frac{\partial}{\partial y_1}+\eta^{(2)}\frac{\partial}{\partial y_2} \right) \left(y_2+\frac{2y_1(y_1+1)+y_1^{3/2}}{x-y} \right) = 0
$$
which implies
$$
-\xi\frac{2y_1(y_1+1)+y_1^{3/2}}{(x-y)^2}+\eta\frac{2y_1(y_1+1)+y_1^{3/2}}{(x-y)^2}+\eta^{(1)}\left(\frac{4y_1+2+\frac{3}{2}y_1^{1/2}}{x-y} \right)+\eta^{(2)} = 0
$$
Recall $$\eta^{(1)}=\eta_x+(\eta_y-\xi_x)y_1+\xi_y y_1^2$$
and  $$\eta^{(2)}=\eta_{xx}+(2\eta_{xy}-\xi_{xx})y_1+(\eta_{yy}-2\xi_{xy})y_1^2-\xi_{yy} y_1^3+(\eta_y-2\xi_x)y_2-2\xi_y y_1y_2$$
Also notice we can replace $y_2$ via the original differential equation, i.e. 
$$
y_2=-\frac{2y_1(y_1+1)+y_1^{3/2}}{x-y} .
$$
Substituting for $\eta^{(1)}$, $\eta^{(2)}$ and $y_2$ via the above expressions, multiplying both sides by $(x-y)^2$ and rearranging the terms as powers of $y_1$ yields the rather long equation:
\begin{eqnarray*}
( 2x\eta_x-2y\eta_x+x^2\eta_{xx}-2xy\eta_{xx}+y^2\eta_{xx} ) && \\ 
+y_1 \left( \begin{array}{l} -2\xi+2\eta+4x\eta_x-4y\eta_x+2x\eta_y-2y\eta_y-2x\xi_x+2y\xi_x \\
+2x^2\eta_{xy}-2xy\eta_{xy}+2y^2\eta_{xy}-x^2\xi_{xx}+2xy\xi_{xx} \\ -y^2\xi_{xx}-2x\eta_y+4x\xi_x+2y\eta_y-4y\xi_x \end{array} \right) && \\
+y_1^2 \left( \begin{array}{l} -2\xi+2\eta+4x\eta_y-4y\eta_y-4x\xi_x+4y\xi_x+2x\xi_y-2y\xi_y \\
+x^2\eta_{yy}-2xy\eta_{yy}+y^2\eta_{yy}-2x^2\xi_{xy} +4xy\xi_{xy} \\ -2y^2\xi_{xy}-2x\eta_y+4x\xi_x+2y\eta_y-4y\xi_x+4x\xi_y-4y\xi_y \end{array} \right) && \\
+y_1^3(4x\xi_y-4y\xi_y-x^2\xi_{yy}+2xy\xi_{yy}-y^2\xi_{yy}+4x\xi_y-4y\xi_y) && \\
+y_1^{1/2}\left(\frac{3}{2}x\eta_x-\frac{3}{2}y\eta_x\right) && \\
+y_1^{3/2}\left(-\xi+\eta+\frac{3}{2}x\eta_y-\frac{3}{2}y\eta_y-\frac{3}{2}x\xi_x+\frac{3}{2}y\xi_x-x\eta_y+2x\xi_x+y\eta_y-2y\xi_x \right) && \\
+y_1^{5/2} \left(\frac{3}{2}x\xi_y-\frac{3}{2}y\xi_y+2x\xi_y-2y\xi_y \right) & = & 0 
\end{eqnarray*}
This expression on the left hand side must vanish, so comparing the coefficients of powers of $y_1$ gives the determining equations:
\begin{eqnarray}
2x\eta_x-2y\eta_x+x^2\eta_{xx}-2xy\eta_{xx}+y^2\eta_{xx} & = & 0 \label{y1zero} \\
-2\xi+2\eta+4x\eta_x-4y\eta_x+2x\eta_y-2y\eta_y-2x\xi_x+2y\xi_x \nonumber \\  
+2x^2\eta_{xy}-2xy\eta_{xy} +2y^2\eta_{xy}-x^2\xi_{xx}+2xy\xi_{xx} \nonumber \\
-y^2\xi_{xx}-2x\eta_y+4x\xi_x+2y\eta_y-4y\xi_x & = & 0 \nonumber \\
-2\xi+2\eta+4x\eta_y-4y\eta_y-4x\xi_x+4y\xi_x+2x\xi_y-2y\xi_y+x^2\eta_{yy} \nonumber \\
-2xy\eta_{yy}+y^2\eta_{yy}-2x^2\xi_{xy}+4xy\xi_{xy}-2y^2\xi_{xy} \nonumber \\
-2x\eta_y+4x\xi_x+2y\eta_y-4y\xi_x+4x\xi_y-4y\xi_y & = & 0 \nonumber \\
4x\xi_y-4y\xi_y-x^2\xi_{yy}+2xy\xi_{yy}-y^2\xi_{yy}+4x\xi_y-4y\xi_y & = & 0 \label{y1three} \\
\frac{3}{2}x\eta_x-\frac{3}{2}y\eta_x & = & 0 \label{y1half} \\
-\xi+\eta+\frac{3}{2}x\eta_y-\frac{3}{2}y\eta_y-\frac{3}{2}x\xi_x+\frac{3}{2}y\xi_x- x\eta_y+2x\xi_x+y\eta_y-2y\xi_x & = & 0 \nonumber \\
\frac{3}{2}x\xi_y-\frac{3}{2}y\xi_y+2x\xi_y-2y\xi_y & = & 0 \label{y1fiveovertwo} 
\end{eqnarray}
It is immediately obvious from $\eqref{y1half}$ that $\eta_x=0$ and from $\eqref{y1fiveovertwo}$ that $\xi_y=0$. Consequently \eqref{y1zero} and \eqref{y1three} vanishes. The remaining determining equations simplify to:
\begin{eqnarray}
-2\xi+2\eta+2x\xi_x-2y\xi_x-x^2\xi_{xx}+2xy\xi_{xx}-y^2\xi_{xx}& = & 0 \label{y1one2} \\
-2\xi+2\eta+2x\eta_y-2y\eta_y+x^2\eta_{yy}-2xy\eta_{yy}
+y^2\eta_{yy} & = & 0 \label{y1two2} \\
-\xi+\eta+\frac{1}{2}x\eta_y-\frac{1}{2}y\eta_y+\frac{1}{2}x\xi_x-\frac{1}{2}y\xi_x & = & 0 \label{y1threeovertwo2}
\end{eqnarray}
These remaining partial differential equations in $\xi(x,y)$ and $\eta(x,y)$ are linear, and recall $\eta(x,y)$ is independent of $x$ and $\xi(x,y)$ is independent of $y$ respectively. 
To solve these partial differential equations we can therefore express $\xi(x)=\sum_{n=0}^{\infty} a_nx^n$ and $\eta(y)=\sum_{m=0}^{\infty} b_my^m$.
Consequently \eqref{y1one2} becomes:
\begin{eqnarray*}
-2\sum_{n=0}^{\infty} a_nx^n+2\sum_{m=0}^{\infty} b_my^m+2\sum_{n=1}^{\infty} na_nx^n-2y\sum_{n=1}^{\infty} na_nx^{n-1} && \\
-\sum_{n=2}^{\infty} n(n-1)a_nx^n+2y\sum_{n=2}^{\infty} n(n-1)a_nx^{n-1}-y^2\sum_{n=2}^{\infty} n(n-1)a_nx^{n-2} & = & 0
\end{eqnarray*}
Comparing the constant term implies $b_0=a_0$.
Comparing the terms containing $y$ implies:
$$2b_1-2\sum_{n=1}^{\infty} na_nx^{n-1}+2\sum_{n=2}^{\infty} n(n-1)a_nx^{n-1}=0$$
Comparing coefficients of $x^n$ gives $b_1=a_1$, $n=n(n-1)$ or $a_n=0$ for $n\geq2$. 
Of course, $n=n(n-1)$ has solution $n=2$ for $n\geq2$, so $a_n=0$ for $n\geq3$. 
Comparing the terms containing $y^2$ implies $b_2=a_2$.
Notice \eqref{y1two2} is symmetric with \eqref{y1one2} in the sense that swapping $\xi$ with $\eta$ and $x$ with $y$ in \eqref{y1one2} gives \eqref{y1two2}. So by symmetry \eqref{y1two2} gives $b_0=a_0$, $b_1=a_1$, $b_2=a_2$ and $b_n=0$ for $n\geq3$.
\eqref{y1threeovertwo2} gives no additional solutions. Therefore, the example ODE admits a three parameter Lie group of transformations with infinitesimals:
$$\xi=a_0+a_1x+a_2x^2$$
$$\eta=a_0+a_1y+a_2y^2$$
where $a_2$, $a_1$ and $a_0$ are arbitrary constants. The infinitesimal generators corresponding to $a_2$, $a_1$ and $a_0$ are respectively
\begin{eqnarray}
\mathrm{X}_1 & = & x^2\frac{\partial}{\partial x}+y^2\frac{\partial}{\partial y} \label{eq: 2nd gen1} \\
\mathrm{X}_2 & = & x\frac{\partial}{\partial x}+y\frac{\partial}{\partial y} \nonumber \\ 
\mathrm{X}_3 & = & \frac{\partial}{\partial x}+\frac{\partial}{\partial y} \label{eq: 2nd gen3} ,
\end{eqnarray}
which generate a three dimensional Lie algebra.

\end{example}

\begin{example}[Ex. \ref{ex: 2nd nonlinear}, continued] \label{ex: 2nd ord reduce}

Recall from Ex. \ref{ex: 2nd nonlinear} that the second order nonlinear ODE in Eq.~\eqref{eq:2ndorderexample} admits a three parameter Lie group of transformations with infinitesimal generators $\mathrm{X}_1$, $\mathrm{X}_2$, $\mathrm{X}_3$ defined in Eqs.~\eqref{eq: 2nd gen1}-\eqref{eq: 2nd gen3}.
These generators can be verified to satisfy $[\mathrm{X}_1,\mathrm{X}_2]=-\mathrm{X}_1$, $[\mathrm{X}_1,\mathrm{X}_3]=-\mathrm{X}_2$, $[\mathrm{X}_2,\mathrm{X}_3]=-\mathrm{X}_3$.
The pairs $\mathrm{X}_1, \mathrm{X}_2$ and $\mathrm{X}_2, \mathrm{X}_3$ form a two dimensional (and therefore solvable by Thm. \ref{thm: 2dsolvable}) Lie sub-algebra and can be used as the basis for our method. 
For the derivations below we proceed with arbitrary choice $\mathrm{X}_1$, $\mathrm{X}_2$.

Following the proof of Thm. \ref{2ndordertoquad}, first we seek a solution $v = v(x,y)$ to the first order linear PDE $\mathrm{X}_1v=0$.
i.e. we must solve
\begin{eqnarray*}
x^2\cfrac{\partial{v}}{\partial x}+y^2\cfrac{\partial{v}}{\partial y}=0
\end{eqnarray*}
This has general solution $v = f(\frac{1}{y}-\frac{1}{x})$ for some arbitrary function $f$, and we pick a particular solution $v(x,y)=\frac{1}{y}-\frac{1}{x}$.
Next we seek a solution $w = w(x,y,y_1)$ to the first order linear PDE $\mathrm{X}_1^{(1)}w=0$.
i.e. we must solve
\begin{eqnarray*}
x^2\frac{\partial w}{\partial x}+y^2\frac{\partial w}{\partial y}+2(y-x) y_1 \frac{\partial w}{\partial y_1} = 0. 
\end{eqnarray*}
Again, we pick a particular solution $w(x,y,y_1)=y_1 (\frac{x}{y})^2$. 
In accordance with Eq.~\eqref{eq:1streduction}, we can re-write the original ODE \eqref{eq:2ndorderexample} in terms of the coordinates $v$ and $w$ to obtain
\begin{eqnarray}\label{eq:2ndorderexamplered1}
 \
    \frac{\mathrm{d}w}{\mathrm{d}v} = \left\{ \begin{array}{lr}
        \cfrac{w^{3/2}+2w(w+1)}{v(w-1)}, & \text{for } \frac{x}{y} \geq 0 \\
        \cfrac{-w^{3/2}+2w(w+1)}{v(w-1)}, & \text{for } \frac{x}{y} < 0
        \end{array}\right.
  \
\end{eqnarray}

Next we express $\mathrm{X}_2^{(1)}$ in terms of $v$ and $w$ find its canonical coordinates $\tilde{r}(v,w)$, $\tilde{s}(v,w)$. 
To this end, we have $\mathrm{X}_2^{(1)}=x\frac{\partial{}}{\partial x}+y\frac{\partial{}}{\partial y}=-v\frac{\partial{}}{\partial v}$, which has canonical coordinates $\tilde{r}(v,w)=w$, $\tilde{s}(v,w)=-\log(v)$.
Re-writing Eq.~\eqref{eq:2ndorderexamplered1} in terms of $\tilde{r}$, $\tilde{s}$ leads to the analogue of Eq.~\eqref{eq: tilde ODE eq}:
\begin{eqnarray}\label{eq:2ndorderexamplered2}
\frac{\mathrm{d}\tilde{\mathrm{s}}}{\mathrm{d}\tilde{r}} \; = \; \frac{1-\tilde{r}}{\pm \tilde{r}^{3/2}+2\tilde{r}(\tilde{r}+1)} \; =: \; H(\tilde{r})
\end{eqnarray}
This example exhibits the convenient feature that Eq.~\eqref{eq:2ndorderexamplered2} can be directly integrated to give $\tilde{\mathrm{s}}(\tilde{r}) =-\log(2\tilde{r}\pm\sqrt{\tilde{r}}+2)+\log(\tilde{r})/2+C$, which can be re-written in terms of $x,y,y_1$ to give 
\begin{eqnarray}
\log \left(\frac{1}{y}-\frac{1}{x} \right) = \log\left( 2 \sqrt{y_1\frac{x^2}{y^2}}  \pm 1 + \frac{2}{\sqrt{y_1\frac{x^2}{y^2}}} \right) - C \label{eq: ex 2nd or rew}
\end{eqnarray}
for some integration constant $C$. 

The final step, to remove the $y_1$ independence, requires canonical coordinates for $\mathrm{X}_1$.
These can be selected as $r(x,y)=\frac{1}{y} - \frac{1}{x}$, $s(x,y)=-\frac{1}{y}$.
Then Eq.~\eqref{eq: ex 2nd or rew} becomes
\begin{eqnarray}
\frac{r\mp\exp(-C)}{2\exp(-C)}=\left(\cfrac{\frac{\mathrm{d}\mathrm{s}}{\mathrm{d}r}}{1+\frac{\mathrm{d}\mathrm{s}}{\mathrm{d}r}}\right)^{1/2}+\left(\cfrac{1+\frac{\mathrm{d}\mathrm{s}}{\mathrm{d}r}}{\frac{\mathrm{d}\mathrm{s}}{\mathrm{d}r}}\right)^{1/2} \nonumber 
\end{eqnarray}
which is equivalent to Eq.~\eqref{eq: thm conc 2nd order} for some function $G$.

\end{example}

\subsection{Design of the Training Set} \label{subsec: training set}

The performance of the proposed Bayesian PNM is not our main focus in this work, as we consider the method to be (only) a proof-of-concept.
However, for completeness we acknowledge that performance will depend on the locations at which the gradient field is evaluated; the so-called \emph{training set}.
In this section we discuss how these inputs could be optimally selected.
To simplify the presentation, we focus on the case of a first order ODE, as in Eq.~\eqref{1storder}, where the inputs $r_0,\dots,r_n$ must be selected.

The design of a PNM can be viewed as an instance of statistical experimental design \citep{Chaloner1995}.
In Sec. 3 of \cite{Cockayne2017} a connection between PNM and decision-theoretic experimental design was exposed.
Such methods require that a loss function $L : \mathcal{Q} \times \mathcal{Q} \rightarrow \mathbb{R}$ is provided, where $L(q,q^\dagger)$ quantifies the loss when $q$ is used as an estimate for the true quantity of interest $q^\dagger$.
Further detail was provided in \cite{Oates2019RICAM}.
To avoid repetition, in the remainder we focus instead on approximate experimental design, where a loss function is not explicitly needed.

Recall that the output of a PNM is the distribution $\mu_n = B(\mu,a^n) \in \mathcal{P}_{\mathcal{Q}}$.
Then one can specify a functional $\ell : \mathcal{P}_{\mathcal{Q}} \rightarrow \mathbb{R}$ and compute
\begin{eqnarray}
\tau(r_0,\dots,r_n) & = & \int \ell(B(\mu,A(\mathrm{y} ; r_0,\dots,r_n))) \mathrm{d}\mu(\mathrm{y}) \label{eq: ell}
\end{eqnarray}
where $A(\cdot ; r_0,\dots,r_n)$ is the information operator in Eq.~\eqref{eq: info operator 1st order} with the dependence on $r_0,\dots,r_n$ made explicit.
For the choice $\ell(\nu) = \log \text{det}(\text{Cov}_{\tilde{Q} \sim \nu}[\tilde{Q}])$, a configuration $(r_0,\dots,r_n)$ for which $\tau(r_0,\dots,r_n)$ is minimised is said to be \emph{D-optimal}.
The functional $\ell$ plays the role of an approximation to posterior expected loss, and other choices for $\ell$ lead to other approximate notions of optimal experimental design.
For instance, an \emph{A-optimal} design was used for the Bayesian solution of a partial differential equation in \cite{Cockayne2016}.
For further background on experimental design we refer the reader to \cite{Chaloner1995}.

Importantly, Eq.~\eqref{eq: ell} does not depend on the information $A(\mathrm{y}^\dagger)$ and can therefore be evaluated prior to the experiment being performed.
However, in general the numerical approximation of Eq.~\eqref{eq: ell}, and the task of finding a minimal configuration, is practically difficult.
The reader is referred to \cite{Overstall2015} for further discussion of experimental design in the PNM context.

\subsection{Computational Detail} \label{subsec: Lopez details}

In this Appendix we set out in detail the prior construction that was used for the numerical illustrations of Sec. \ref{sec: experiment} in the main text.
Recall that for both the first order ODE example in Sec. \ref{subsec: 1st order experiment} and the second order ODE example in Sec. \ref{subsec: 2nd order experiment} we required a non-parametric prior over functions $\zeta$ which satisfy the constraints given in Eq. \ref{eq: prior criterion 1st order}, namely that
\begin{eqnarray}
\zeta(r_0)=0, \qquad \zeta(r) \leq 1, \qquad \frac{\mathrm{d}\zeta}{\mathrm{d}r} \geq 0  .
\label{eq: prior criterion 1st order, again}
\end{eqnarray}
Moreover, bearing in mind the posterior computation that is to follow, we require in addition that the prior conveniently facilitates the conditioning calculations involved.
It is clear that standard non-parametric priors such as Gaussian processes do not satisfy the boundedness or monotonicity constraints, whilst a nonlinear transformation of such a process would fail to make conditioning on data straight-forward.
In fact, the construction of such flexible priors remains an active area of research.

To proceed, we adopted an approach recently proposed in \cite{Lopez-Lopera2017}.
In brief, the main idea is to construct an $N$-dimensional parametric distribution over functions for which Eq. \ref{eq: prior criterion 1st order, again} is satisfied.
This distribution, being finite-dimensional, allows for the possibility of tractable conditioning operations, whilst the flexibility to take $N$ arbitrarily large provides a means of ensuring that the salient uncertainty is accurately represented.
More specifically, the function $\zeta$ is parametrised as
\begin{eqnarray}
\zeta(r)=\sum_{j=1}^{N} z_j \phi_j(r)
\label{eq: seriespriorton}
\end{eqnarray}
where the $\phi_j$ are basis functions
$$
\phi_j(r)=
\begin{cases} 
      1-|\frac{r-t_j}{h}| & |\frac{r-t_j}{h}|\leq1 \\
      0 & \mathrm{otherwise} 
\end{cases}
$$
for equally spaced points $t_j$ with increment $h$, as recommended in \cite{Lopez-Lopera2017}. 
A prior on $\zeta$ can be induced via a prior on the coefficients $z_1,\dots,z_N$, with $N$ taken to be substantially larger than the number $n$ of datapoints on which the $z_i$ are to be conditioned.
The specific construction of a prior on the coefficients is required to encode the constraints in Eq. \ref{eq: prior criterion 1st order, again} and to admit tractable conditioning; these issues are discussed in the remainder.

First, we consider the boundedness and monotonicity constraints in Eq. \ref{eq: prior criterion 1st order, again}.
At the level of the coefficients, it is straight-forward to check that this requires that the prior support is restricted to the set
$$\mathcal{Z}=\{z \in \mathbb{R}^N: 0 < z_1 \leq z_2 \leq \dots \leq z_N \leq 1\}.$$
For convenience, we elected to use a prior that was obtained by restricting a standard Gaussian measure $\mathcal{N}(0,I)$ on $\mathbb{R}^N$ to the set $\mathcal{Z}$.

Second, we consider how to condition on a dataset.
Recall that information is provided on the values of the gradient $\zeta'(r_i) = b_i$ of the function $\zeta$, evaluated at a finite number of locations $r_i$ of the canonical coordinate $r$, together with the initial condition $\zeta(r_0)=b_0$ .
Thus the information can be described by the linear system of constraints
$$\Phi z=b$$
where 
\[\Phi=
\begin{bmatrix}
    \phi_1(r_0) & \dots  & \phi_N(r_0) \\
    \phi'_1(r_1) & \dots  & \phi'_N(r_1)  \\
    \vdots  & \vdots \\
    \phi'_1(r_n) & \dots  & \phi'_N(r_n)
\end{bmatrix}, \qquad b = \begin{bmatrix} b_0 \\ b_1 \\ \vdots \\ b_n \end{bmatrix} .
\]
The posterior can therefore be characterised as the restriction of $\mathcal{N}(0,I)$ to the set $\mathcal{Z} \cap \mathcal{D}$ where $\mathcal{D} = \{z \in \mathbb{R}^N : \Phi z = b\}$.

Finally, we discuss how posterior computation was performed.
The key observation is that an equivalent characterisation of the posterior is first to restrict $\mathcal{N}(0,I)$ to $\mathcal{D}$ and then to further restrict to $\mathcal{Z}$.
This is advantageous since the linear nature of the data implies that the restriction $z | \mathcal{D}$ of $\mathcal{N}(0,I)$ to $\mathcal{D}$ is again a Gaussian with a closed form, denoted $\mathcal{N}(\mu,\Sigma)$.
It is important to note that $\Sigma$ is singular (rank $\rho=N-n-1$) and so $\Sigma=U\Lambda^2 U^\top$, where $U$ is an orthogonal matrix and $\Lambda$ is a diagonal matrix with $\rho$ non-zero entries on the diagonal. 
Thus we can express $z|\mathcal{D}$ in the form 
$$z=\mu+U\Lambda \tilde{z}$$
where $\tilde{z} \sim \mathcal{N}(0,I)$ is a standard Gaussian on $\mathbb{R}^\rho$. 
Let $M=U\Lambda$ and let $m_i=[m_{i,1}, \dots m_{i,\rho}]$ denote the $i$th row of $M$. 
Then we have the relation
$$z \in \mathcal{Z} \iff \tilde{z} \in \tilde{\mathcal{Z}}$$
where 
$$\tilde{\mathcal{Z}}=\{\tilde{z}  \in \mathbb{R}^\rho: 0\leq m_1\tilde{z}+\mu_1, 0 \leq (m_{i+1}-m_i)\tilde{z}+(\mu_{i+1}-\mu_i), 0 \leq -m_N\tilde{z}+(1-\mu_N)\}$$
or equivalently
$$\tilde{\mathcal{Z}}=\{\tilde{z}  \in \mathbb{R}^\rho: F\tilde{z}+g \geq 0\}$$
for
\[F=
\begin{bmatrix}
    m_1 \\
    m_2-m_1  \\
    \vdots  \\
    m_N-m_{N-1} \\
    -m_N
\end{bmatrix},
\quad
g=
\begin{bmatrix}
    \mu_1 \\
    \mu_2-\mu_1  \\
    \vdots  \\
    \mu_N-\mu_{N-1} \\
    1-\mu_N
\end{bmatrix} .
\]
The computational task is thus reduced to sampling the restriction of the $\rho$-dimensional standard Gaussian random variable $\tilde{z}$ to the (non-null) set $\tilde{\mathcal{Z}}$.
The development of computational methods to sample from such (potentially high-dimensional) distributions is itself an active area of research, and for this work we employed the Hamiltonian Monte Carlo method of \cite{Pakman2014}, as recommended specifically for this purpose in \cite{Lopez-Lopera2017}.

\color{black}

\end{document}